\documentclass[journal]{IEEEtran}
\newcounter{mytempeqncnt}
\newtheorem{thm}{ Theorem}
\newtheorem{lem}{Lemma} 

\newenvironment{proof}[1][Proof]{\begin{trivlist}
\item[\hskip \labelsep {\bfseries #1}]}{\end{trivlist}}

\newcommand{\qed}{\nobreak \ifvmode \relax \else
      \ifdim\lastskip<1.5em \hskip-\lastskip
      \hskip1.5em plus0em minus0.5em \fi \nobreak
      \vrule height0.75em width0.5em depth0.25em\fi}

\ifCLASSINFOpdf
  \usepackage[pdftex]{graphicx}
  \graphicspath{{./}{Figs/}}
  \DeclareGraphicsExtensions{.pdf,.jpeg,.png}
\else
\fi
\usepackage{amssymb}
\usepackage{amsmath}
\usepackage{mathtools}
\usepackage{bigdelim}
\usepackage{caption}
\usepackage{subcaption}

\newcommand{\MeijerG}[7]{G^{#1,#2}_{#3,#4} \left(\! \begin{smallmatrix} #5 \\ #6 \end{smallmatrix} \middle\vert #7 \!\right) }

\IEEEoverridecommandlockouts

\begin{document}
\title{Random Aerial Beamforming for Underlay Cognitive Radio with Exposed Secondary Users}
\author{Ahmed~M.~Alaa,~\IEEEmembership{Student Member,~IEEE}, Mahmoud~H.~Ismail,~\IEEEmembership{Member,~IEEE} and~Hazim~Tawfik
\thanks{The authors are with the Department
of Electronics and Electrical Communications Engineering, Faculty of Engineering, Cairo University, Giza 12613, Egypt (e-mail: \{aalaa, mismail, htawfik\}@eece.cu.edu.eg).}
\thanks{Manuscript received XXXX XX, 201X; revised XXXX XX, 201X.}}

\markboth{XXXXXXX,~Vol.~XX, No.~X, XXXX~201X}
{Alaa \MakeLowercase{\textit{et al.}}: Random Aerial Beamforming for Underlay Cognitive Radio with Exposed Secondary Users}

\maketitle
\begin{abstract}
In this paper, we introduce the {\it exposed secondary users problem} in {\it underlay cognitive radio} systems, where both the secondary-to-primary and primary-to-secondary channels have a Line-of-Sight (LoS) component. Based on a Rician model for the LoS channels, we show, analytically and numerically, that LoS interference hinders the achievable secondary user capacity when interference constraints are imposed at the primary user receiver. This is caused by the poor dynamic range of the interference channels fluctuations when a dominant LoS component exists. In order to improve the capacity of such system, we propose the usage of an Electronically Steerable Parasitic Array Radiator (ESPAR) antennas at the secondary terminals. An ESPAR antenna involves a single RF chain and has a reconfigurable radiation pattern that is controlled by assigning arbitrary weights to $M$ orthonormal basis radiation patterns via altering a set of reactive loads. By viewing the orthonormal patterns as multiple {\it virtual dumb antennas}, we randomly vary their weights over time creating artificial channel fluctuations that can perfectly eliminate the undesired impact of LoS interference. This scheme is termed as {\it Random Aerial Beamforming} (RAB), and is well suited for compact and low cost mobile terminals as it uses a single RF chain. Moreover, we investigate the exposed secondary users problem in a multiuser setting, showing that LoS interference hinders multiuser interference diversity and affects the growth rate of the SU capacity as a function of the number of users. Using RAB, we show that LoS interference can actually be exploited to improve multiuser diversity via {\it opportunistic nulling}.           	
\end{abstract}
\IEEEpeerreviewmaketitle

\IEEEpeerreviewmaketitle
\section{Introduction}
\IEEEPARstart{S}{ignificant} interest has recently been devoted to the capacity analysis of underlay cognitive radio systems in fading environments. In underlay cognitive radio, a Secondary User (SU) aggressively transmits its data over the Primary User (PU) channel while keeping the interference experienced by the PU below a predefined {\it interference temperature} \cite{1}. The PU is usually assumed to be oblivious to the SU activity, thus power control is applied by the SU transmitter in order to meet with the predefined interference constraints. In \cite{2}, Gastpar has shown that imposing a receive power constraint in an AWGN interference channel does not affect the channel capacity. However, in a fading environment, Ghasemi and Sousa proved that channel fluctuations can be exploited to improve the SU capacity when the Channel State Information (CSI) is available at the SU transmitter \cite{2d}. This capacity improvement is attributed to the ability of the SU to transmit with very high power when the interference channel is severely faded. In \cite{3} and \cite{4}, Musavian {\it et al.} derived the ergodic, outage, and minimum-rate capacities under peak and average interference constraints at the PU receiver. However, the interference experienced by the SU receiver due to primary transmission was not considered. Recently, the ergodic capacity of underlay cognitive radio taking both PU and SU interference into consideration was calculated in \cite{5}. However, the analysis therein is limited to the case when all channels are severely faded, and considers an average interference power constraint only.

In a practical underlay cognitive setting, the SU capacity generally depends on two interference channels, namely; the primary-to-secondary and secondary-to-primary channels \cite{5}. If these channels are subject to Non-Line-of-Sight (NLoS) fading, then the interference channel gains perceived by the PU and SU receivers fluctuate drastically over time \cite{6}. Therefore, the SU transmitter can exploit such fluctuations by opportunistically allocating higher power to time instants when the Signal-to-Interference-and-Noise-Ratio (SINR) is large, and lower power to time instants with poor SINR \cite{2}-\cite{5}. In this paper, we study a practical underlay cognitive radio system with joint peak and average interference power constraints, where the mutual interference channels have dominant LoS (specular) components. We show that in this case, there are limited opportunities for the SU due to the poor dynamic range of channel fluctuations. This results in a significant SU capacity degradation. We term this problem as the {\it exposed secondary user problem}, i.e., the SU is exposed to the PU via a direct LoS link. Such problem would arise in emerging cognitive radio technologies such as underlay {\it Device-to-Device (D2D) Communications}, where the small antenna heights imply that a strong LoS component is likely to exist between the primary and secondary devices \cite{7}. It also appears in other recent cognitive radio approaches, such as the application of underlay cognitive radio to satellite systems \cite{10x3}, where the Rician channel models inherent in conventional satellite communications will indeed be involved.

A straightforward approach to improve the capacity of an arbitrary cognitive radio scheme is the deployment of multiple antennas. Traditional multiple antenna diversity techniques were employed in \cite{8}-\cite{10x1} to improve the SU capacity in spectrum sharing systems. However, the usage of multiple uncorrelated antennas is inhibited by the space limitations of mobile SU transceivers. This is in addition to the need for multiple RF chains, which increases the cost and complexity of the SU equipment. While such overhead is bearable for a base station, it can not be tolerated for modern mobile terminals. Moreover, emerging cellular underlay {\it D2D} technology involves a mobile SU transmitter and a mobile SU receiver \cite{11}, which prevents the deployment of multiple antennas at either terminals. In order to reduce the hardware complexity of multiple antenna systems, underlay cognitive radio with antenna selection is implemented using a single RF chain in \cite{9} and \cite{10x1}. However, this scheme requires antennas to be sufficiently seperated in order to achieve full spatial diversity order. Moreover, an important limitation of all multiple-antenna schemes is the essence of obtaining the CSI for every antenna in order to achieve diversity or multiplexing gains.

The main contribution of this paper is the usage of single Electronically Steerable Parasitic Array Radiator (ESPAR) antennas at both the SU transmitter and receiver. In LoS channels, these antennas are used to create artificial channel fluctuations to restore the transmission opportunities. This is achieved by a technique that we refer to as {\it Random Aerial Beamforming} (RAB), where random time-varying complex weights are assigned to the orthonormal basis radiation patterns of the ESPAR antenna. Inspired by the seminal work of Viswanath \textit{et al.} \cite{6}, we term the orthonormal radiation patterns provided by the ESPAR antenna as {\it dumb basis patterns}, since they represent Degrees of Freedom (DoFs) that are neither used to achieve diversity nor to realize multiplexing. Our analytical and numerical results show that the proposed scheme can make the LoS interference transparent to both the secondary and primary systems, thus achieving the same capacity of the symmetric Rayleigh channel. Furthermore, if the secondary-to-secondary channel has a LoS component, we present a novel scheme that induces fluctuations in the interference channels while maintaining the reliability of the secondary-to-secondary channel. We refer to this technique as {\it artificial diversity using smart basis patterns}. The proposed schemes have several advantages. First, no extra hardware complexity is involved as only a single RF chain is used. Second, only overall CSI is required to apply optimal power allocation. Finally, the proposed schemes can be used in low cost mobile terminals with tight space limitations.

In addition to the capacity benefits of RAB in point-to-point underlay cognitive radio systems, we show that RAB can play an important role in multiuser systems. By generalizing the exposed SUs problem to an $N$ user {\it parallel access channel}, we prove that LoS interference reduces the scaling law of the SU sum capacity from $\log(N)$ to $\log(\log(N))$. By using RAB, artificial fluctuations are induced in the interference channels restoring the $\log(N)$ scaling law. A fundamental result of this paper is that using RAB, LoS interference can actually be exploited. We show that using only 2 basis patterns, multiuser interference diversity can be improved via opportunistic nulling. In this case, LoS interference act as a friend and not a foe.

The application of ESPAR and other classes of reconfigurable antennas to wireless communications systems is not actually new. In \cite{12}-\cite{23}, the {\it Aerial DoF (ADoF)} provided by the orthonormal basis patterns are used to construct single radio {\it Beamspace-Multi-Input Multi-Output (MIMO)} systems that can apply spatial multiplexing without the need to deploy multiple antennas. Recently, beamspace-MIMO was applied for millimeter-Wave (mm-Wave) systems using parasitic antenna arrays with reconfigurable loads \cite{10x7}. Reconfigurable antennas with predefined sets of radiation patterns were also employed in \cite{10x5} to achieve diversity in spectrum sensing for interweave cognitive radio. In \cite{16}, ESPAR antennas were exploited for detecting spatial holes in multiple PU cognitive radio systems. Moreover, blind interference alignment was implemented utilizing the reconfigurability feature of the ESPAR antenna in \cite{17}. Other applications of ESPAR antennas in multiuser systems can be found in \cite{18}.

The rest of the paper is organized as follows: the system model is presented in Section II. In Section III, we present the exposed secondary users problem highlighting the capacity degradation due to LoS interference. The concept of RAB is then proposed in Section IV. In Section V, we present the multiple exposed secondary users problem and investigate the impact of RAB on multiuser interference diversity. Numerical results are presented in Section V and conclusions are drawn in Section VI.

\section{System Model}
We divide this section into three subsections. First, we briefly review the ESPAR antenna structure and parameters. Next, we present the beamspace domain channel model and show that orthonormal basis patterns can be viewed as virtual multiple antennas. Finally, we present a generic system model for the spectrum-sharing system under study with and without ESPAR antennas at SU terminals.
\begin{figure}[t]
\centering
\includegraphics[width=2.5in]{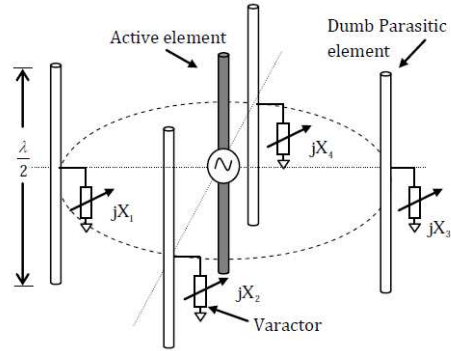}
\caption{The ESPAR antenna with 4 parasitic elements.}
\label{fig_sim}
\end{figure}

\subsection{The ESPAR Antenna}
As shown in Fig. 1, an ESPAR with $M$ elements is composed of a single active element (e.g., a $\frac{\lambda}{2}$ dipole) that is surrounded by $M-1$ identical parasitic elements. Unlike multi-antenna systems, the parasitic elements are placed relatively close to the active elements. Hence, mutual coupling between different elements takes place and current is induced in all parasitic elements. The radiation pattern of the ESPAR is altered by tuning a set of $M-1$ reactive loads (varactors) $\mathbf{x} = \left[jX_{1} \ldots jX_{M-1}\right]$ attached to the parasitic elements \cite{13}. The currents in the parasitic and active elements are represented by an $M \times 1$ vector $\mathbf{i} = v_{s} (\mathbf{Y}^{-1}+\mathbf{X})^{-1}\mathbf{u}$, where
\[
\mathbf{Y} =
 \begin{pmatrix}
  y_{11} & y_{12} & \cdots & y_{1M} \\
  y_{21} & y_{22} & \cdots & y_{2M} \\
  \vdots  & \vdots  & \ddots & \vdots  \\
  y_{M1} & y_{M2} & \cdots & y_{MM}
 \end{pmatrix}\]
is the $M \times M$ admittance matrix with $y_{ij}$ being the mutual admittance between the $i^{th}$ and $j^{th}$ elements. The load matrix $\mathbf{X}$ = {\bf diag}$\left(50,\,\, \mathbf{x}\right)$\footnote{The opertaion {\bf X = diag(x)} embeds a vector {\bf x} in the diagonal matrix {\bf X}.} controls the ESPAR beamforming, $\mathbf{u} = \left[1 \,\, 0 \ldots 0\right]^{T}$ is an $M \times 1$ vector and $v_{s}$ is the complex feeding at the active element \cite{13}. The radiation pattern of the ESPAR at an angle $\theta$ is thus given by $P(\theta) = \mathbf{i}^{T}\mathbf{a}(\theta)$, where $\mathbf{a}(\theta) = \left[a_{0}(\theta) \ldots a_{M-1}(\theta)\right]^{T}$ is the steering vector of the ESPAR at an angle $\theta$ given by \cite{13}-\cite{23}
\[a_{m}(\theta) = \left\{
  \begin{array}{lr}
   1, \,\,\,\,\,\,\,\,\,\,\,\, m = 0\\
    \mbox{exp}(jb \, \mbox{cos}(\theta-\theta_{m})), m = 1, ..., M-1
  \end{array}
\right.\]
where $\theta_{m} = 2\pi(m-1)/(M-1), \,\, m = 1,...,M-1, \,\, b = 2\pi\frac{d}{\lambda}$, and $d$ is the radius of the circular arrangement of parasitic antenna elements. The beamspace domain is a signal space where any radiation pattern can be represented as a point in this space. To represent the radiation pattern $P(\theta)$ in the beamspace domain, the steering vector $\mathbf{a}(\theta)$ is decomposed into a linear combination of a set of orthonormal basis patterns $\{\Phi_{i}(\theta)\}_{i=0}^{N-1}$ using Gram-Schmidt orthonormalization, where $N \leq M$ \cite{16}. It can be shown that the orthonormal basis patterns of the ESPAR (also known as the ADoF \cite{13}-\cite{23}) are equal to the number of parasitic elements (i.e., $N = M$). Therefore, the ESPAR radiation pattern in terms of the orthonormal basis patterns can be written as \cite{13}
\begin{equation}
\label{1}
P(\theta) = \sum_{n=0}^{M-1} w_{n} \Phi_{n}(\theta),
\end{equation}
where $w_{n} = \mathbf{i}^{T}\mathbf{q}_{n}$ are the weights assigned to the basis patterns and $\mathbf{q}_{n}$ is an $M \times 1$ vector of projections of all the steering vectors on $\Phi_{n}(\theta)$. Thus, the ESPAR radiation pattern is formed by manipulating the reactive loads attached to the parasitic elements.

\subsection{Orthonormal basis patterns as multiple virtual antennas}

\begin{figure}[t]
\centering
\includegraphics[width=3.5 in]{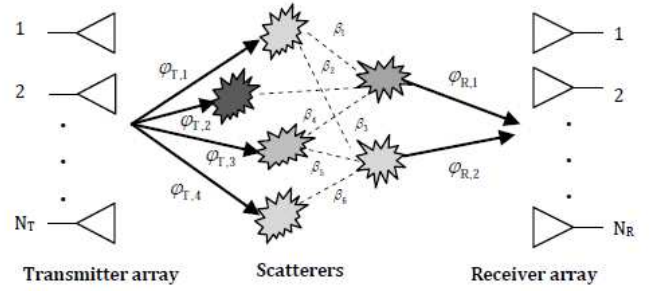}
\caption{Conventional MIMO spatial multipath channel.}
\label{fig_sim}
\end{figure}

\begin{figure}[t]
\centering
\includegraphics[width=3.5 in]{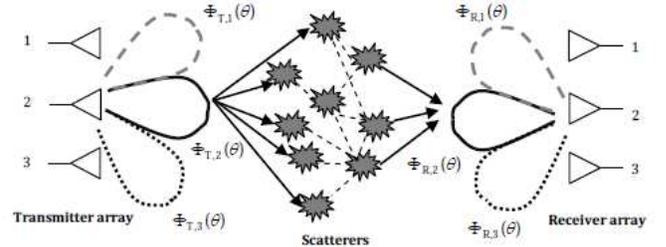}
\caption{Representation of a 3-antenna MIMO system in the beamspace domain.}
\label{fig_sim}
\end{figure}
In this section, we construct a generic channel model for the ESPAR antenna in the beamspace domain showing that the orthonormal basis patterns offer Aerial Degrees of Freedom (ADoF) and can act as multiple virtual antennas. Assume a conventional MIMO system with $N_{T}$ transmit antennas and $N_{R}$ receive antennas. The channel between the transmitter and the receiver is characterized by the presence of $Q$ scatterers as shown in Fig. 2. The $q^{th}$ scattering path is represented by a fading gain of $\beta_{q}$, and a unique pair of transmit and receive angles $\phi_{T,q}$ and $\phi_{R,q}$. In this case, the received signal vector ${\bf r} \in \mathbb{C}^{N_{R} \times 1}$ is given by
\[{\bf r} = {\bf H} \, {\bf s} + {\bf n},\]
where ${\bf H} \in \mathbb{C}^{N_{R} \times N_{T}}$ is the channel matrix representing the coupling between antenna elements, ${\bf s} \in \mathbb{C}^{N_{T} \times 1}$ is the transmitted signal vector, and ${\bf n} \in \mathbb{C}^{N_{R} \times 1}$ is the noise vector. The channel matrix can be formulated as \cite{12}
\begin{equation}
{\bf H} = \sum_{q=1}^{Q} \beta_{q} {\bf a_{R}}(\theta_{R,q}) {\bf a_{T}}^{H}(\theta_{T,q}) = {\bf A_{R}(\theta_{R})} {\bf H_{b}} {\bf A_{T}}^{H}{\bf (\theta_{T})},
\end{equation}
where ${\bf a_{R}}(\theta_{R,q})$ and ${\bf a_{T}}(\theta_{T,q})$ are the steering vectors of the transmitting and receiving arrays, ${\bf A_{R}(\theta_{R})}$ and ${\bf A_{T}(\theta_{T})}$ are the receive and transmit steering matrices where ${\bf A_{R}(\theta_{R})} = [{\bf a_{R}}(\theta_{R,1}),...,{\bf a_{R}}(\theta_{R,Q})]_{N_{R} \times Q}$ and ${\bf A_{T}(\theta_{T})} = [{\bf a_{T}}(\theta_{T,1}),...,{\bf a_{T}}(\theta_{T,Q})]_{N_{T} \times Q}$, and $H_{b} = \mbox{diag}(\beta_{1},...,\beta_{Q})$ is a $Q \times Q$ diagonal matrix. An equivalent model for the MIMO channel is the {\it virtual channel model} \cite{11x1}, which represents the coupling between virtual angular directions at the transmitter and the receiver. The channel matrix ${\bf H}$ in this model is given by (dropping the dependency on the angles for convenience) 
\begin{equation}
{\bf H} = {\bf \tilde{A}_{R} \, H_{v}} \, {\bf\tilde{A}_{T}}^{H},
\end{equation}
where $\tilde{A}_{R} \in \mathbb{C}^{N_{R} \times N_{R}}$ and $\tilde{A}_{T} \in \mathbb{C}^{N_{T} \times N_{T}}$ are the steering matrices corresponding to the virtual directions. Iff $\tilde{A}_{R}$ and $\tilde{A}_{T}$ are unitary, then the virtual channel matrix can be written as,
\begin{equation}
{\bf H_{v}} = {\bf \tilde{A}_{R}}^{H} {\bf A_{R} \, H_{b} \,} {\bf A_{T}}^{H} {\bf \tilde{A}_{T}}.
\end{equation}
The matrices ${\bf \tilde{A}_{R}}^{H} {\bf A_{R}}$ and ${\bf \tilde{A}_{T}}^{H} {\bf A_{T}}$ are the projections of ${\bf A_{R}}$ and ${\bf A_{T}}$ onto the square steering matrices corresponding to the virtual directions. It was shown in \cite{12} that ${\bf \tilde{A}_{R}}^{H} {\bf A_{R}}$ is a $Q \times N_{R}$ matrix containing vectors that represent the response of the orthonormal basis patterns to the scatterers. Thus, ${\bf \tilde{A}_{R}}^{H} {\bf A_{R}} = {\bf \Phi_{R}}$ and similarly, ${\bf \tilde{A}_{T}}^{H} {\bf A_{T}} = {\bf \Phi_{T}}$, where ${\bf \Phi_{R}}^{H} = [\Phi_{R,1}, \Phi_{R,2},...,\Phi_{R,N_{R}}] \in \mathbb{C}^{Q \times N_{R}}$ and ${\bf \Phi_{T}}^{H} = [\Phi_{T,1}, \Phi_{T,2},...,\Phi_{T,N_{T}}] \in \mathbb{C}^{Q \times N_{T}}$ contain the values of the $N_{R}$ and $N_{T}$ transmit and receive basis patterns towards the direction of the scatterers. The system input-output relationship can thus be written as
\begin{eqnarray*}
{\bf r} &=& {\bf H} \, {\bf s} + {\bf n}
\\ &=& {\bf \tilde{A}_{R} \, H_{v} \,} {\bf \tilde{A}_{T}}^{H} {\bf s} + {\bf n}
\\  &=& {\bf \tilde{A}_{R} \Phi_{R} \, H_{b} \,} {\bf \Phi_{T}}^{H} {\bf \tilde{A}_{T}}^{H} {\bf s} + {\bf n}.
\end{eqnarray*}
Assume that the data symbols are sent over diverse orthogonal beams rather than antenna elements, in this case the input to the antenna elements is not equal to the data symbols. The input to the antenna elements is related to the actual transmitted data symbols $s_{bs}$ via a precoding matrix, i.e.,  $s_{bs} = {\bf \tilde{A}_{T}}^{H} \, {\bf s}$ and the received signal perceived by the orthogonal beams at the receiver is given by $r_{bs} = {\bf \tilde{A}_{R}}^{H} \, {\bf r}$. Thus, an alternative implementation for MIMO transmission is to send and receive symbols in the beamspace domain, and the channel matrix in this case can be represented as \cite{18}
\begin{equation}
\label{5}
{\bf H_{bs}} = {\bf \Phi_{R}} \, {\bf H_{b}} \, {\bf \Phi_{T}}^{H}.
\end{equation}
The matrix ${\bf H_{bs}}$ represents the coupling between the orthonormal basis patterns of the transmitter antenna array and the corresponding patterns at the receiver. The ADoF, which denote the amount of diverse dimensions in beamspace on which independent streams can be loaded, is given by $\mbox{rank}({\bf H_{bs}})$. In Fig. 3, a 3 $\times$ 3 MIMO system is depicted. The conventional implementation of MIMO transmission is to load every antenna element with an independent data symbol. Alternatively, the implementation of MIMO in the beamspace domain is to load each of the transmit orthonormal basis patterns ($\Phi_{T,1}, \Phi_{T,2}, \Phi_{T,3}$) with an independent symbol by appropriate precoding and applying the techniques in \cite{23} to demultiplex the symbols perceived by the receive basis patterns.

Now, recalling the structure of the ESPAR antenna, we note that the concept of its operation is identical to beamspace MIMO. Although the ESPAR antenna has only one active element, it can provide DoF as it is capable of assigning independent weights to orthonormal patterns via adjusting reactive loads. Thus, an ESPAR antenna is analogous to MIMO systems in the beamspace domain, and we can indeed consider the set of orthonormal patterns as multiple virtual antennas.

\subsection{Underlay Cognitive Radio System Model}
\begin{figure}[t]
\centering
\includegraphics[width=3.5 in]{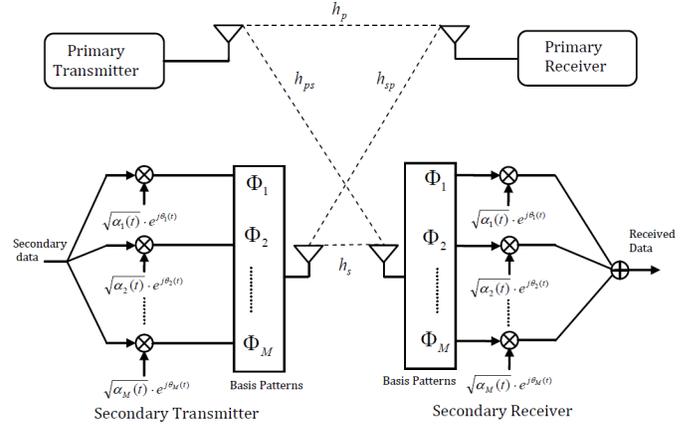}
\caption{Spectrum sharing system model.}
\label{fig_sim}
\end{figure}
Assume an underlay cognitive radio system where single-user primary and secondary systems coexist as shown in Fig. 4. The received signals at the $k^{th}$ time instant are given by
\begin{equation}
\label{2}
\begin{array}{lcl} \mbox{At PU-Rx:} \,\,\, r_{p}(k) = h_{p}(k) x_{p}(k) + h_{sp}(k) x_{s}(k) + n_{p}(k),\\ \mbox{At SU-Rx:} \,\,\, r_{s}(k) = h_{s}(k) x_{s}(k) + h_{ps}(k) x_{p}(k) + n_{s}(k), \end{array}
\end{equation}
 where $h_{p}(k)$ (PU-to-PU), $h_{ps}(k)$ (PU-to-SU), $h_{sp}(k)$ (SU-to-PU) and $h_{s}(k)$ (SU-to-SU) are the respective complex-valued overall channel gains (linear combination of channel gains at all basis patterns). The primary and secondary signals $x_{p}(k)$ and $x_{s}(k)$ are complex-valued symbols drawn from an $L$-ary constellation, while $n_{s}(k)$ and $n_{p}(k)$ are the AWGN samples with power spectral density $N_{o}$ at the secondary and primary receivers, respectively. The PU terminals are assumed to be equipped with single omnidirectional antennas while the antennas provided at the SU terminals are assumed to have $M$ orthonormal basis patterns having a complex weight vector $\mathbf{w} = \mathbf{i}^{T}\mathbf{q}$, where the $n^{th}$ weight $w_{n}$ has an arbitrary complex value $\sqrt{\alpha_{n}} e^{j \theta_{n}}$. The weight value depends on the setting of the parasitic elements reactive loads. This model can be reduced to the conventional single antenna (with no parasitic elements) by setting the weight vector to $\mathbf{w} = \left[1 \,\,\, 0 \ldots 0\right]$, in which case one ADoF is available, corresponding to the active antenna element. We also assume perfect knowledge of these channels at the SU transmitter and receiver. It is worth mentioning that our analysis fits any design for reconfigurable antennas with beamforming capabilities, and not only the ESPAR antenna\footnote{A comprehensive framework for single-radio reconfigurable antenna design and analysis can be found in \cite{23}.}.

\begin{figure}[!t]
\centering
\includegraphics[width=3.5 in]{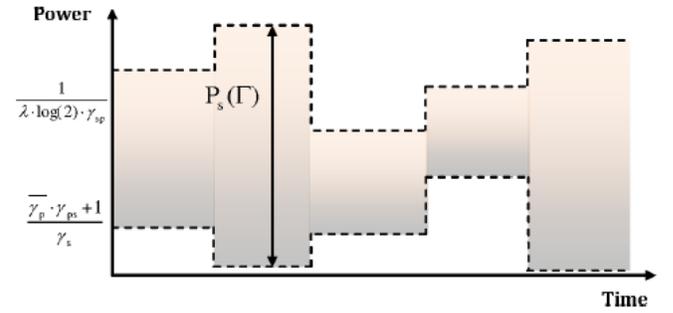}
\caption{Non-uniform water filling power allocation.}
\end{figure}

\section{The Exposed Secondary Users Problem}
We start with the conventional single antenna scheme with all parties employing single omnidirectional antennas. In this case, we set the weight vector for the model in Section II to $\mathbf{w} = \left[1 \,\,\, 0 \ldots 0\right]$.

\subsection{Ergodic Capacity Formulation}
The ergodic capacity maximization problem is a generalization for the problems in \cite{2}, \cite{3} and \cite{5}. We adopt a joint peak and average interference power constraints. In addition, we extend the work in \cite{2} and \cite{3} by considering the PU-to-SU as well as the SU-to-PU interference. The problem can be formulated as
\begin{align}
\label{3}
 C &= \max_{\mathbf{P_{s}}(\mathbf{\Gamma})} \mathbb{E}_{\mathbf{\Gamma}}\left\{\log_{2}\left(1+\frac{\gamma_{s}\mathbf{P_{s}}(\mathbf{\Gamma})}{\gamma_{ps}\overline{\gamma}_{p} + N_{o}}\right)\right\},\nonumber\\
&\mbox{subject to} \,\,\,\, \mathbb{E}_{\mathbf{\Gamma}} \left\{\gamma_{sp}\mathbf{P_{s}}(\mathbf{\Gamma})\right\} \leq Q_{av},\\
& \,\,\,\,\,\mbox{and} \,\,\,\, {\bf P_{s}}(\mathbf{\Gamma}) \gamma_{sp} \leq Q_{p},\nonumber
\end{align}
where $\mathbf{\Gamma} = (\gamma_{s},\gamma_{sp},\gamma_{ps}) = (|h_{s}|^{2},|h_{sp}|^{2},|h_{ps}|^{2})$, $Q_{av}$ and $Q_{p}$ are the average and peak interference power constraints, respectively, $\mathbb{E}\{.\}$ is the expectation operator, and ${\bf P_{s}}(\mathbf{\Gamma})$ is the SU transmit power allocation as a function of the CSI vector $\mathbf{\Gamma}$. Without loss of generality, we set the noise spectral density $N_{o}$ = 1 W/Hz. Thus the SNR of any link is equal to the transmit power of the PU/SU transmitter multiplied by the channel power. The PU is assumed to transmit with a constant power of $\overline{\gamma}_{p}$. The optimization problem in (\ref{3}) can be easily solved using the Lagrangian method \cite{3}. We define the Lagrangian function $\Xi(\mathbf{P_{s}}(\mathbf{\Gamma}), \lambda)$ as
\[\Xi(\mathbf{P_{s}}(\mathbf{\Gamma}), \lambda) = \mathbb{E}_{\mathbf{\Gamma}}\left\{\log_{2}\left(1+\frac{\gamma_{s}\mathbf{P_{s}}(\mathbf{\Gamma})}{\gamma_{ps}\overline{\gamma}_{p} + 1}\right)\right\} \]
\[
- \lambda \left(\mathbb{E}_{\mathbf{\Gamma}} \left\{\gamma_{sp}\mathbf{P_{s}}(\mathbf{\Gamma})\right\} - Q_{av}\right),
 \]
where $\lambda$ is the Lagrange multiplier. The solution for the optimal power allocation is obtained by maximizing the Lagrangian function as
\[
\frac{\partial \, \Xi(\mathbf{P_{s}}(\mathbf{\Gamma}), \lambda)}{\partial \mathbf{P_{s}}(\mathbf{\Gamma})} = 0,
\]
 yielding $\mathbf{P_{s}}(\mathbf{\Gamma}) = \frac{1}{\lambda \gamma_{sp} \log(2)}-\frac{\gamma_{ps}\overline{\gamma}_{p}+1}{\gamma_{s}}$. In order to satisfy the peak constraint ${\bf P_{s}}(\mathbf{\Gamma}) \gamma_{sp} \leq Q_{p}$, and the constraint ${\bf P_{s}}(\mathbf{\Gamma}) \geq 0$, we have
\begin{equation}
\label{4}
{\bf P_{s}}(\mathbf{\Gamma}) = \left\{
  \begin{array}{lr}
   \frac{Q_{p}}{\gamma_{sp}}, \,\,\,\,\,\,\,\,\,\,\,\, \frac{\gamma_{s}}{\gamma_{sp}} \leq \frac{(1+\gamma_{ps}\overline{\gamma}_{p})}{\frac{1}{\lambda \log(2)}-Q_{p}}\\
    \frac{1}{\lambda \gamma_{sp} \log(2)}-\frac{\gamma_{ps}\overline{\gamma}_{p}+1}{\gamma_{s}},\\ \,\,\,\,\,\,\,\,\,\,\,\,\,\,\,\,\,\,\,\,\,\,\,\, \frac{(1+\gamma_{ps}\overline{\gamma}_{p})}{\frac{1}{\lambda \log(2)}-Q_{p}} \leq \frac{\gamma_{s}}{\gamma_{sp}} \leq \lambda \log(2) (1+\gamma_{ps}\overline{\gamma}_{p}) \\
		0, \,\,\,\,\,\,\,\, \frac{\gamma_{s}}{\gamma_{sp}} \geq \lambda \log(2) (1+\gamma_{ps}\overline{\gamma}_{p})
  \end{array}
\right.
\end{equation}
where the Lagrange multiplier $\lambda$ is obtained numerically to satisfy the constraints in (\ref{3}). Note that the peak interference constraint $Q_{p}$ must be greater than the average interference constraint $Q_{av}$. Thus, the peak constraint has an impact on SU capacity only if $Q_{p} > \frac{1}{\lambda \log(2)}$. If the average constraint is very large (i.e., $\lambda$ is very small), then defining a peak interference constraint is meaningless.
By relaxing the peak interference constraint, we obtain the oprimal power allocation as ${\bf P_{s}}(\mathbf{\Gamma}) = \left\{\frac{1}{\lambda \gamma_{sp} \log(2)}-\frac{\gamma_{ps}\overline{\gamma}_{p}+1}{\gamma_{s}}\right\}^{+},$ where $\{x\}^{+} = \max\{x,0\}$. This power allocation policy can be graphically interpreted as water filling with a non-uniform (time varying) water level as shown in Fig. 5. The water level depends on the SU-to-PU fading channel and is quantified by the term $\frac{1}{\lambda \gamma_{sp} \log(2)}$. The allocated power at each time instant is the difference between the {\it water level} $\frac{1}{\lambda \gamma_{sp} \log(2)}$ and the {\it water depth} $\frac{\gamma_{ps}\overline{\gamma}_{p}+1}{\gamma_{s}}$. The optimal power allocation policy is sensitive to the $\gamma_{s}$, $\gamma_{sp}$, and $\gamma_{ps}$ channels statistics as follows:
\begin{itemize}
  \item The water level is controlled by the reciprocal of $\gamma_{sp}$. If $\gamma_{sp}$ is not severely faded, then the problem reduces to the conventional constant water level power allocation policy. If $\gamma_{sp}$ is severely faded, the water level fluctuates over time allowing for larger amount of power to be allocated.
  \item The water depth is controlled by both $\gamma_{ps}$ and the reciprocal of $\gamma_{s}$. If $\gamma_{s}$ is severely faded, then the water depth is relatively small because $\frac{1}{\gamma_{s}}$ is likely to have large values most of the time. Contrarily, if $\gamma_{ps}$ is severely faded, the water depth is relatively large because $\gamma_{ps}$ is likely to have small values.
\end{itemize}

Based on (\ref{4}), the ergodic capacity is given in (\ref{5}), where $z = \frac{\gamma_{s}}{\gamma_{sp}}$ and $f_{z}(z)$ is the probability density function (pdf) of $z$. The random variable $z$ jointly describes the water level and depth, i.e., it is proportional to the ratio between both quantities. When $\gamma_{s}$ is severely faded and $\gamma_{sp}$ has a LoS component, the random variable $z$ will have an exponentially decaying tail, and it will be highly likely that $z$ has a small value. On the other hand, when $\gamma_{sp}$ is severely faded and $\gamma_{s}$ has a LoS component, the random variable $z$ will have a heavier tail and it will be highly likely that $z$ will have large values, which corresponds to large allocated power.

In the following subsections, we investigate the SU capacity when the SUs are exposed to the PU via a direct LoS channel. We highlight the negative impact of LoS interference on the SU capacity and show that LoS interference hinders the capacity gain obtained by opportunistic power allocation. In the following analysis, an $X-Y$ scenario describes the one where the SU-to-SU channel is subject to a fading distribution $Y$ and the mutual interference channels (SU-to-PU and PU-to-SU channels) are subject to fading distribution $X$. In each scenario, we aim at identifying the behavior of the interference channels and the power allocation opportunities offered to the SU transmitter.

\begin{figure*}[!t]
\normalsize
\setcounter{mytempeqncnt}{\value{equation}}
\setcounter{equation}{8}
\begin{equation}
\label{5}
C = \mathbb{E}_{\gamma_{ps}} \left\{ \int_{z = \lambda \log(2) (1+\gamma_{ps}\overline{\gamma}_{p})}^{\frac{(1+\gamma_{ps}\overline{\gamma}_{p})}{\left\{\frac{1}{\lambda \log(2)}-Q_{p}\right\}^{+}}} \log_{2}\left(\frac{z}{\lambda \log(2)(\gamma_{ps}\overline{\gamma}_{p}+1)}\right) f_{z}(z) dz + \int_{z = \frac{(1+\gamma_{ps}\overline{\gamma}_{p})}{\left\{\frac{1}{\lambda \log(2)}-Q_{p}\right\}^{+}}}^{\infty} \log_{2}\left(1+\frac{Q_{p}\,z}{(\gamma_{ps}\overline{\gamma}_{p}+1)}\right) f_{z}(z) dz \right\}.
\end{equation}
\setcounter{equation}{\value{mytempeqncnt}+1}
\hrulefill
\vspace*{4pt}
\end{figure*}

\subsection{The Rician-Rician Scenario}
In this scenario, all the channels have LoS components. In other words, both the SU-to-SU and the mutual interference channels are not severely faded. We adopt a Rician model for the SU-to-SU, SU-to-PU, and PU-to-SU channels as follows:
\begin{equation}
\label{6}
h_{i}(k) = \sqrt{\overline{\gamma}_{i}}\left(\sqrt{\frac{K_{i}}{K_{i}+1}} e^{j \phi_{i}} + v_{i}(k)\right)
\end{equation}
where $h_{i}(k) \in \{h_{s}(k),h_{sp}(k),h_{ps}(k)\}$, $K_{i} \in \{K_{s},K_{sp},K_{ps}\}$, $v_{i}(k) \in \{v_{s}(k),v_{sp}(k),v_{ps}(k)\}$, and $\phi_{i} \in \{\phi_{s},\phi_{sp},\phi_{ps}\}$. The parameters $K_{s}$, $K_{sp}$, and $K_{ps}$ are the $K$-factors (specular components) of the $h_{s}$, $h_{sp}$, and $h_{ps}$ channels, respectively. The components $v_{s}(k)$, $v_{sp}(k)$, and $v_{ps}(k)$ are the diffused components where $\left(v_{s}(k),v_{sp}(k),v_{ps}(k)\right) \sim \mathcal{CN}\left(0,\left(\frac{1}{K_{s}+1}, \frac{1}{K_{sp}+1}, \frac{1}{K_{ps}+1}\right)\right)$, while ${\phi_{s}, \phi_{sp}, \phi_{ps}}$ are the constant phases of the LoS components. Non-identical Rician channels with average channel powers of $(\overline{\gamma}_{s},\overline{\gamma}_{ps},\overline{\gamma}_{sp})$ are assumed for the channels $(h_{s},h_{ps},h_{sp})$. The pdf of the channel power $\gamma_{i} = |h_{i}|^{2}$ when $|h_{i}|$ follows a Rician distribution is given by \cite{19}
\begin{equation}
\label{7}
f_{\gamma_{i}}(\gamma_{i}) = \frac{1+K_{i}}{\overline{\gamma}_{i}} e^{-K_{i}-\frac{(1+K_{i})}{\overline{\gamma}_{i}} \gamma_{i}} I_{o}\left(2 \sqrt{\frac{K_{i}(1+K_{i})}{\overline{\gamma}_{i}}\gamma_{i}}\right)
\end{equation}
where $I_{o}(.)$ is the modified Bessel function of the first kind. In the Rician-Rician scenario, both $\gamma_{sp}$ and $\gamma_{ps}$ follow the distribution in (\ref{7}). The dynamic range of the interference channels ($|h_{sp}|$,$|h_{ps}|$) fluctuations can be expressed by the variance ($\sigma^{2}_{sp,ps}$) of the Rician distribution \cite{19}
\[ \sigma^{2}_{sp,ps} =\]
\[ 2\frac{\overline{\gamma}_{sp,ps}}{K_{sp,ps}+1} + \frac{K_{sp,ps}}{K_{sp,ps}+1} - \frac{\pi \overline{\gamma}_{sp,ps}}{2(K_{sp,ps}+1)} L_{1/2}^{2}\left(\frac{-K_{sp,ps}}{2\overline{\gamma}_{sp,ps}}\right),\]
where $L_{1/2}^{2}(.)$ is the square of the Laguerre polynomial. An interesting scenario is when the $K$-factor tends to $\infty$. Given that the limit of the Laguerre polynomial is $\lim_{x \to -\infty} L_{v}(x) = \frac{|x|^{v}}{\Gamma(v+1)}$, where $\Gamma(.)$ is the gamma function \cite{19}, it can be easily shown that the variance of the interference channels $\sigma^{2}_{sp,ps}$ tends to 0 when $K_{sp,ps} \to \infty$. Thus, a dominant LoS interference signal makes the interference channel almost deterministic. Therefore, for the Rician-Rician scenario with $K_{s,sp,ps} \to \infty$, the pdfs of $z$ and $\gamma_{ps}$ tend to
\[f_{z}(z) \to \delta\left(z-\frac{\overline{\gamma}_{s}}{\overline{\gamma}_{sp}}\right), \,\,
\mbox{and} \,\,
f_{\gamma_{ps}}(\gamma_{ps}) \to \delta(\gamma_{ps}-\overline{\gamma}_{ps}).\]
The ergodic SU capacity is obtained by plugging $f_{z}(z)$ and $f_{\gamma_{ps}}(\gamma_{ps})$ in (\ref{5}), leading to the same capacity of the AWGN interference channel under receive power constraint as follows
\[\lim_{K_{sp,ps} \to \infty} C = \log_{2}\left(1+\frac{\frac{Q_{av}}{\overline{\gamma}_{sp}} \overline{\gamma}_{s}}{1+\overline{\gamma}_{p} \overline{\gamma}_{ps}}\right),\]
which is the SU capacity when all channels are AWGN with no fading. Based on this analysis, we conclude the following remark.\\
{\bf Remark 1:} Because no SU transmit power constraint is imposed, the amount of average power transmitted by the SU is controlled by the statistics of $\gamma_{sp}$ and $\gamma_{ps}$. If both channels are not severely faded, the SU transmits with an almost constant power level, which can be interpreted as an equivalent transmit power constraint whose value is decided by $\overline{\gamma}_{sp}$ and $\overline{\gamma}_{ps}$. This means that non-severely faded interference channels hinder the opportunistic behavior of the SU. Furthermore, the existence of a LoS component in the SU-to-SU channel makes the SU capacity reduce to the deterministic AWGN capacity.

\subsection{The Rician-Rayleigh Scenario}
In this scenario, the SU-to-SU channel is characterized by a large dynamic range (large variance) and follows a Rayleigh distribution, while the interference channels have a LoS component and follow a Rician distribution. The random variable $z = \frac{\gamma_{s}}{\gamma_{sp}}$ is the ratio between the squares of a Rayleigh and a Rician random variable. In order to obtain the capacity, we first rewrite the pdf of $\gamma_{sp}$ in (\ref{7}) in terms of the Meijer-$G$ function $\MeijerG{m}{n}{p}{q}{a_1,\ldots,a_p}{b_1,\ldots,b_q}{z}$ [26, Sec. 7.8] as
\[f_{\gamma_{sp}}(\gamma_{sp}) =\]
\begin{equation}
\label{8}
 \frac{1+K_{sp}}{\overline{\gamma}_{sp}} e^{-K_{sp}-\frac{(1+K_{sp})\gamma_{sp}}{\overline{\gamma}_{sp}}} \MeijerG{1}{0}{0}{2}{-}{0, 0}{\frac{K_{sp} (1+K_{sp}) \gamma_{sp}}{\overline{\gamma}_{sp}}}.
\end{equation}
Noting that $\gamma_{s}$ follows an exponential distribution with $f_{\gamma_{s}}(\gamma_{s}) = \frac{1}{\overline{\gamma}_{s}} e^{-\frac{\gamma_{s}}{\overline{\gamma}_{s}}}$ and that the pdf of $Z = \frac{X}{Y}$ is given by $f_{Z}(z) = \int_{-\infty}^{+\infty}|y| p_{x,y}(zy, y) dy$ \cite{19}, one can obtain the expression in (\ref{9}) for the pdf of $z$. Using the property $z \MeijerG{m}{n}{p}{q}{a_1,\ldots,a_p}{b_1,\ldots,b_q}{z} = \MeijerG{m}{n}{p}{q}{a_1+1,\ldots,a_p+1}{b_1+1,\ldots,b_q+1}{z}$, the integral can be evaluated as the standard laplace transform of a Meijer-$G$ function \cite{20} yielding the result in (16). We notice that when the $K$-factors $\to\infty$, the pdfs of interest tend to
\[\lim_{K_{sp} \to \infty} f_{z}(z) = \frac{\overline{\gamma}_{sp}}{\overline{\gamma}_{s}} e^{\frac{-\overline{\gamma}_{sp}z}{\overline{\gamma}_{s}}} ,  \,\, \lim_{K_{ps} \to \infty} f_{\gamma_{ps}}(\gamma_{ps}) = \delta(\gamma_{ps} - \overline{\gamma}_{ps}).\]
Thus, the ergodic capacity of the SU is obtained by plugging the pdfs $f_{z}(z)$ and $f_{\gamma_{ps}}(\gamma_{ps})$ in (\ref{5}). Unlike the Rician-Rician scenario, the variable $z$ is not deterministic at large values of $K_{sp}$. Instead, $z$ is exponentially distributed, which implies that $f_{z}(z)$ has an exponentially-bounded tail. Because a pdf with an exponentially-decaying tail gives small weight to large values of $z$, the capacity of the SU is still limited as it depends on the integral $\int \log(z)f_{z}(z) dz$. We quantify the SU capacity of this scenario in the following lemma.

\begin{lem}
{\it For asymptotically large $K$-factors and infinite peak interference constraint, the SU capacity in the Rician-Rayleigh Scenario is bounded by
\[C \to {\rm Ei} \left(-\lambda \log(2) (\gamma_{ps} \overline{\gamma}_{p}+1) \frac{\overline{\gamma}_{sp}}{\overline{\gamma}_{s}}\right).\]
where ${\rm Ei}(.)$ is the exponential integral function. At large SNR for the SU-to-SU channel ($\overline{\gamma}_{s} \to \infty$), the SU capacity tends to
\[C  \lesssim \log_{2}\left(\frac{\overline{\gamma}_{s}}{\lambda_{AWGN} \log(2) \overline{\gamma}_{sp} (\overline{\gamma}_{p} \gamma_{p} + 1)}\right) - \gamma,\]
where $\lesssim$ is the asymptotic inequality and $\gamma$ is the Euler-Mascheroni constant. At low SU-to-SU SNR ($\overline{\gamma}_{s} \to 0$), the SU capacity tends to
\[C \gtrsim \frac{\overline{\gamma}_{s} \exp\left(-\frac{\lambda_{AWGN} \log(2) \overline{\gamma}_{sp} (\overline{\gamma}_{p} \gamma_{p} + 1)}{\overline{\gamma}_{s}}\right)}{\lambda_{AWGN} \log(2) \overline{\gamma}_{sp} (\overline{\gamma}_{p} \gamma_{p} + 1)},\]
where $\lambda_{AWGN}$ is the lagrange multiplier for the AWGN channel.
}
\end{lem}
\begin{proof} See Appendix A. \IEEEQEDhere
\end{proof}
Note that the SU capacity in terms of the Lagrangian multiplier when all the channels are AWGN is given by $C_{AWGN} = \log_{2}\left(\frac{\overline{\gamma}_{s}}{\lambda_{AWGN} \log(2) \overline{\gamma}_{sp} (\overline{\gamma}_{p} \gamma_{p} + 1)}\right)$. Thus, the results of Lemma 1 suggest the following:
\begin{itemize}
  \item For an asymptotically large $\overline{\gamma}_{s}$, the SU capacity in the Rician-Rayleigh scenario is less than the AWGN capacity by a constant gap that is equal to the {\it Euler-Mascheroni} constant $\gamma \approx 0.577$ bits/sec/Hz.
  \item For asymptotically small $\overline{\gamma}_{s}$, the SU capacity in the Rician-Rayleigh scenario is greater than the AWGN capacity as it can be shown that $\frac{1}{x} e^{-x} > \log(\frac{1}{x}), \forall x > 0$.
\end{itemize}
These conclusions lead to the following remark.\\

{\bf Remark 2:} Severe fading in the SU-to-SU link is generally a source of {\it unreliability}. In other words, SU-to-SU channel fluctuations results in an unreliable channel. The only way to exploit such fluctuations is to optimally allocate power based on the SU transmitter CSI at low SNR, where the capacity is inherently power limited. This corresponds to the conventional point-to-point fading channel with CSI at the transmitter analyzed by Goldsmith and Varaiya in \cite{21}. Moreover, the LoS components in the interference channels hinder any opportunities to improve the capacity, and the optimal power allocation is almost oblivious to the interference channel fluctuations.

\begin{figure*}[!t]
\normalsize
\setcounter{mytempeqncnt}{\value{equation}}
\setcounter{equation}{12}
\begin{align}
\label{9}
f_{z}(z) &= \frac{1+K_{sp}}{\overline{\gamma}_{sp} \overline{\gamma}_{s}} \,\,\, e^{-K_{sp}} \int_{0}^{\infty} \gamma_{sp} e^{-\left(\frac{(1+K_{sp})}{\overline{\gamma}_{sp}}+\frac{z}{ \overline{\gamma}_{s}}\right) \gamma_{sp}} \MeijerG{1}{0}{0}{2}{-}{0, 0}{\frac{K_{sp} (1+K_{sp})\gamma_{sp}}{\overline{\gamma}_{sp}}} d \gamma_{sp} \\
 &= \frac{2 e^{-K_{sp}}}{K_{sp} \overline{\gamma}_{s}} \int_{0}^{\infty} e^{-\left(\frac{(1+K_{sp})}{\overline{\gamma}_{sp}}+\frac{z}{ \overline{\gamma}_{s}}\right) \gamma_{sp}} \MeijerG{1}{0}{0}{2}{-}{1, 1}{\frac{K_{sp} (1+K_{sp})\gamma_{sp}}{\overline{\gamma}_{sp}}} d \gamma_{sp} \\
&= \frac{2 e^{-K_{sp}}}{K_{sp} \overline{\gamma}_{s}} \times \frac{1}{\frac{1+K_{sp}}{\overline{\gamma}_{sp}}+\frac{z}{\overline{\gamma}_{s}}} \times \MeijerG{1}{1}{1}{2}{0}{1, 1}{\frac{K_{sp} (1+K_{sp})}{\overline{\gamma}_{sp}\left(\frac{(1+K_{sp})}{\overline{\gamma}_{sp}}+\frac{z}{\overline{\gamma}_{s}}\right)}} \\
&= \frac{(1+K_{sp}) \overline{\gamma}_{sp}}{\overline{\gamma}_{s}} \,\,\, e^{-K_{sp}+\frac{K_{sp}(1+K_{sp})}{(1+K_{sp})+\frac{\overline{\gamma}_{sp}}{\overline{\gamma}_{s}}z}} \left(\frac{(1+K_{sp})^{2}+z \frac{\overline{\gamma}_{sp}}{\overline{\gamma}_{s}}}{\left((1+K_{sp})+z \frac{\overline{\gamma}_{sp}}{\overline{\gamma}_{s}}\right)^{3}}\right).
\end{align}
\setcounter{equation}{\value{mytempeqncnt}+4}
\hrulefill
\vspace*{4pt}
\end{figure*}

\subsection{The Rayleigh-Rayleigh Scenario}
In this scenario, all channels are severely faded with drastic fluctuations over time. Because $h_{sp}$ does not have a LoS component, the pdf of $z$ can be obtained by setting $K_{sp}$ to 0 in (16) yielding
\[f_{z}(z) = \frac{\overline{\gamma}_{sp}}{\overline{\gamma}_{s}\left(1+z \frac{\overline{\gamma}_{sp}}{\overline{\gamma}_{s}}\right)^{2}},\]
which is the {\it log-logistic} distribution. It can be shown that the log-logistic distribution is a {\it fat-tailed distribution} by showing that $P[Z>z] \sim z^{-\alpha}$ as $z \to \infty$, where $\alpha > 0$ and $\sim$ is the asymptotic equivalence. Thus, the pdf of the variable $z$ has a heavier right-tail in the Rayleigh-Rayleigh scenario compared to the Rician-Rician and Rician-Rayleigh scenarios, which means that $\int \log(z)f_{z}(z)$ in (\ref{5}) will be larger in the Rayleigh-Rayleigh scenario as $f_{z}(z)$ is slowly decaying and will give significantly larger weights to larger values of $\log(z)$ (note that $\log(z)$ is a monotonically increasing function of $z$). On the other hand, $\gamma_{ps}$ still has an exponentially bounded tail. Again, this contributes to the SU capacity enhancement as $C$ depends on $-\int_{\gamma_{ps}=0}^{\infty}\log(\gamma_{ps}+v)f_{\gamma_{ps}}(\gamma_{ps})$ in (\ref{5}), where $v$ is a constant. Thus, severely faded interference channels offer SU capacity gain even if the SU-to-SU channel is severely faded as well. The SU opportunistic behavior can be quantified by the pdf tails for $z$ and $\gamma_{ps}$. The heavier the tail of $f_{z}(z)$ and the faster the pdf of $\gamma_{ps}$ decays, the larger is the SU capacity. Recalling the water filling power allocation, the heavy tail of $z$ means that the ratio between the water level and depth is likely to have a large value (i.e., large allocated power).

\subsection{The Rayleigh-Rician Scenario}
In this scenario, the SU-to-SU channel is a reliable LoS channel, while the interference channels are severely faded. Assuming that the SU-to-SU $K$-factor is asymptotically large, we can approximate the pdf of $\gamma_{s}$ as $f_{\gamma_{s}}(\gamma_{s}) \approx \delta(\gamma_{s}-\overline{\gamma}_{s})$. Each of the interference channels $\gamma_{sp}$ and $\gamma_{ps}$ follows an exponential distribution. Thus, the variable $z = \frac{\gamma_{s}}{\gamma_{sp}}$ is the reciprocal of an exponential random variable. The pdf of $z$ will be given by a simple random variable transformation as
\[\lim_{K_{s} \to \infty} f_{z}(z) = \frac{\overline{\gamma}_{s}}{\overline{\gamma}_{sp} \, z^{2}} e^{\frac{-\overline{\gamma}_{s}}{\overline{\gamma}_{sp} \, z}}.\]
Note that $f_{z}(z) \sim z^{-2}$ when $z \to \infty$, which means that the distribution of $z$ is a fat-tailed distribution. The distribution of $\gamma_{ps}$ has an exponentially distributed tail as $h_{ps}$ follows a Rayleigh distribution. These results lead to the following remark.\\

{\bf Remark 3:} When the interference channels are severely faded, and the SU-to-SU channel has a LoS component with an asymptotically large $K$-factor, $z$ is characteized by fat-tailed distribution and $\gamma_{ps}$ follows an exponential distribution with an exponentially decaying tail. Thus, the water level to depth ratio is likely to have large values over time allowing for opportunistic power allocation at the SU transmitter.

Combining remarks 1, 2, and 3, we reach the following conclusion. The capacity of an underlay cognitive system with interference constraints at the PU receiver depends on the statistics of the PU-to-SU and SU-to-PU interference channels. When the interference channels are severely faded, the water filling power allocation can offer a significant capacity gain compared to the AWGN scenario. However, when the SUs are {\it exposed} to the PUs via LoS channels, the SU transmitter finds limited power allocation opportunities and the capacity of such system becomes similar to the AWGN capacity for asymptotically large $K$-factors. The impact of the SU-to-SU channel statistics is similar to conventional point-to-point fading channels, where fading is generally considered a source of unreliability. However, at low SU-to-SU channel SNR, such fluctuations can be exploited to improve capacity when optimal power allocation is applied.
		
\section{Random Aerial Beamforming using Dumb Basis Patterns}
Motivated by the analysis in the previous section, we propose a technique that can eliminate the impact of LoS interference by improving the dynamic range of interference channels fluctuations. This is achieved by RAB, which intentionally induces artificial fluctuations in these channels by randomizing the complex weights assigned to the basis patterns of an ESPAR antenna. Throughout this section, we adopt the system model presented in Section II, with an ESPAR antenna weight vector of $\mathbf{w} = \mathbf{i}^{T}\mathbf{q} = \left[\begin{array}{c} \sqrt{\alpha_{1}} e^{j \theta_{1}} \,\, \ldots \,\, \sqrt{\alpha_{M}} e^{j \theta_{M}}\end{array}\right]$.

\subsection{Random Aerial Beamforming (RAB)}
By viewing the ADoF provided by the orthonormal basis patterns as {\it virtual dumb antennas} or {\it dumb basis patterns}, we adjust the reactive loads of the parasitic antenna elements such that the weights assigned to the basis patterns are randomly varied over time. Based on the system model presented in Section II, an ESPAR antenna with $M-1$ parasitic elements employed at the transmitter and the receiver has an input-output relationship given by ${\bf r_{bs}} = {\bf H_{bs}} \, {\bf s_{bs}} + {\bf n},$ where ${\bf s_{bs}} \in \mathbb{C}^{M \times 1}$ is the set of weights assigned to the basis patterns obtained by altering the reactive loads, ${\bf r_{bs}}$ is the signal vector received at the beamspace domain, and ${\bf n}$ is a noise vector. The received beamspace vector can be further reduced to
\begin{equation}
\label{33}
{\bf r_{bs}} = {\bf \Phi_{R}} \, {\bf H_{b}} \, {\bf \Phi_{T}}^{H} \, {\bf s_{bs}} + {\bf n},
\end{equation}
where ${\bf \Phi_{R}}^{H} = [\Phi_{R,1}, \Phi_{R,2},...,\Phi_{R,M}] \in \mathbb{C}^{Q \times M}$ and ${\bf \Phi_{T}}^{H} = [\Phi_{T,1}, \Phi_{T,2},...,\Phi_{T,M}] \in \mathbb{C}^{Q \times M}$ are the responses of the transmit and receive basis patterns towards the $Q$ scatterers. If a LoS component with a specific $K$-factor exists between the transmitter and the receiver, the channel model reduces to
\begin{equation}
\label{35}
{\bf r_{bs}} = \left(\sqrt{\frac{K}{K+1}} {\bf \overline{H}_{bs}} + \frac{1}{\sqrt{K+1}} {\bf \Phi_{R}} \, {\bf H_{b}} \, {\bf \Phi_{T}}^{H} \right) \, {\bf s_{bs}} + {\bf n},
\end{equation}
where ${\bf \overline{H}_{bs}}$ is the beamspace deterministic LoS channel matrix with entries $[\sqrt{\overline{\gamma}} \, e^{j \phi^{k,m}}]$, where $k=0,1,..., M-1$, $m = 0, 1, ..., M-1$, and $\overline{\gamma}$ is the average SNR. The entries of ${\bf \overline{H}_{bs}}$ denote the deterministic phase shifts of the LoS components at various basis patterns. The implementation of RAB in the SU receiver is different from its implementation at the transmitter as explained hereunder.

\subsubsection{RAB at the SU transmitter}
The goal of applying RAB at the SU transmitter is to induce artificial fluctuations in $h_{sp}$. Assume a SU transmitter with $M_{T}$ basis patterns ($M_{T}-1$ parasitic antenna elements). For the SU transmitter to send a symbol $x_{s}(k)$ to the SU receiver at time instant $k$, it selects a basis patterns weight vector $\mathbf{s_{bs}} = \left[x_{s}(k)\sqrt{\alpha_{T,1}(k)} e^{j \theta_{T,1}(k)} \,\, \ldots \,\, x_{s}(k) \sqrt{\alpha_{T,M_{T}}(k)} e^{j \theta_{T,M_{T}}(k)}\right]$. In this case, ${\bf s_{bs}} = {\bf w_{T}}(k) x(k)$, where ${\bf w_{T}}(k) \in \mathbb{C}^{M_{T} \times 1}$.
Without loss of generality, we set $\alpha_{T,i}(k) = \frac{1}{M_{T}}, \forall i$ and randomly vary the phases $\theta_{T,i}(k)$ every time instant $k$ according to a uniform distribution $\theta_{i}(k) \sim \mbox{Unif}(0,2\pi)$ (independent phases are selected for all basis patterns). Hence, at each time instant $k$, the SU transmitter adjusts the reactive loads such that $\mathbf{s_{bs}} = \left[\frac{x_{s}(k) e^{j \theta_{T,1}(k)}}{\sqrt{M_{T}}} \,\, \ldots \,\, \frac{x_{s}(k) e^{j \theta_{T, M_{T}}(k)}}{\sqrt{M_{T}}}\right]$.

\subsubsection{RAB at the SU receiver}
In order to induce fluctuations in the PU-to-SU interference channel $h_{ps}$, RAB must be employed by the SU receiver. The receiver uses an ESPAR antenna with $M_{R}$ basis patterns. A weight vector of $\mathbf{w_{R}}(k) = \left[ \frac{e^{j \theta_{R,1}(k)}}{\sqrt{M_{R}}} \,\, \ldots \,\, \frac{e^{j \theta_{R,M_{R}}(k)}}{\sqrt{M_{R}}}\right]$, where the set of phase shifts $\theta_{R,i}(k), i = 1,...,M_{R},$ are independent and uniformly distributed. Because an ESPAR antenna has a single RF chain, the received beamspace signal vector ${\bf r_{bs}}$ can not be handeled in the RF front-end. In fact, the RF chain perceives a linear combination of the elements of ${\bf r_{bs}}$ weighted by ${\bf w_{R}}(k)$. We denote the signal received at the SU receiver as $r(k) = {\bf w_{R}}^{H}(k) \, {\bf r_{bs}}$.

\subsubsection{Equivalent RAB Signal Model}
The signal model at the $k^{th}$ time instant for a transmitter-receiver pair using ESPAR antennas and applying RAB can be obtained by setting $r(k) = {\bf w_{R}}^{H}(k) \, {\bf r_{bs}}$ and ${\bf s_{bs}} = {\bf w_{T}}(k) \, x(k)$ thus yielding the result in (\ref{37}). We consider $\Phi_{T,i} = \left[\Phi_{T,i}(\varphi_{1}(k)), \, \Phi_{T,i}(\varphi_{2}(k)) \ldots \,\,, \Phi_{T,i}(\varphi_{Q}(k))\right]^{T}$ to be a set of $Q$ spatial samples of the $i^{th}$ basis pattern towards the $Q$ scatterers at the $k^{th}$ instant, where the $q^{th}$ element is the spatial element at a departure angle of $\varphi_{q}$. The same applies to $\Phi_{R,i}$ for the receive basis patterns. The time varying gain of the $q^{th}$ scatterer is denoted by $\beta_{q}(k)$. from this equation, it is clear that applying RAB results in an equivalent channel $h^{eq}(k)$, with new statistics other than the Rician distribution. The statistics of $h^{eq}(k)$ depend on the weight vectors ${\bf w_{R}}(k)$ and ${\bf w_{T}}(k)$. Revisiting the signal model in (\ref{2}), we rewrite the received signals at the PU and SU receivers ($r_{p}(k)$ and $r_{s}(k)$, respectively) in (21) and (22). We assume that the SU transmitter and receiver employ ESPAR antennas with $M_{T}$ and $M_{R}$  basis patterns, and response matrices ${\bf \Phi_{T}}_{M_{T} \times Q}$ and ${\bf \Phi_{R}}_{M_{R} \times Q}$, respectively while the PU transmitter and receiver use conventional antennas with radiation patterns ${\bf \tilde{\Phi}_{T}} \in \mathbb{C}^{1 \times Q}$ and ${\bf \tilde{\Phi}_{R}} \in \mathbb{C}^{1 \times Q}$, and weighting vectors ${\bf \tilde{w}_{R}}(k)$ and ${\bf \tilde{w}_{T}}(k)$, respectively. Without loss of generality, we set  ${\bf \tilde{w}_{R}}(k) = {\bf 1}$ and ${\bf \tilde{w}_{T}}(k) = {\bf 1}$. Let $i \in \{s, sp, ps\}$ denote the SU-to-SU, SU-to-PU, and PU-to-SU channels, respectively . Channel $i$ is characterized by the following parameters: a $K$-factor $K_{i}$, a deterministic LoS beamspace channel matrix ${\bf \overline{H}^{i}_{bs}}$ with entries $[\sqrt{\overline{\gamma}^{i}} \, e^{j\phi^{k,m}_{i}}]$ denoting the phase of the LoS component between the $k^{th}$ transmit basis pattern and the $m^{th}$ receive basis pattern, $\overline{\gamma}^{i}$ is the average SNR, and $\beta^{i}_{q}$ is the gain of the $q^{th}$ scattered component. $\beta^{i}_{q}$ is distributed with an arbitrary pdf with a variance of $\overline{\gamma}^{i}$. In the following theorem, we derive the statistics of the equivalent channels $h^{eq}_{s}(k)$, $h^{eq}_{sp}(k)$, and $h^{eq}_{ps}(k)$ after applying RAB.

\begin{thm}
{\it For large enough number of transmit and receive basis patterns $M_{T}$ and $M_{R}$, if the SU-to-SU, SU-to-PU, and PU-to-SU channels have LoS components with $K$-factors of $K_{s}$, $K_{sp}$, and $K_{ps}$ , respectively , the equivalent channels after applying RAB are Rayleigh distributed. Thus, $h^{eq}_{s}(k) \sim \mathcal{CN}(0,\overline{\gamma}_{s})$, $h^{eq}_{sp}(k) \sim \mathcal{CN}(0,\overline{\gamma}_{sp})$, and $h^{eq}_{ps}(k) \sim \mathcal{CN}(0,\overline{\gamma}_{ps})$.}
\end{thm}

\begin{proof} See Appendix B. \IEEEQEDhere
\end{proof}

Theorem 1 implies that RAB restores the opportunities hindered by LoS interference by inducing fluctuations in the Rician channels. However, inducing fluctuations in the SU-to-SU channel is not preferrable at high SU-to-SU average SNR as clarified by Lemma 1. In the next subsection, we present a variation of RAB that preserves the reliability of the SU-to-SU channel.

\begin{figure*}[!t]
\normalsize
\setcounter{mytempeqncnt}{\value{equation}}
\setcounter{equation}{18}
\begin{align}
\label{37}
r(k) &= \underbrace{{\bf w_{R}}^{H}(k) \left(\sqrt{\frac{K}{K+1}} {\bf \overline{H}_{bs}} + \frac{1}{\sqrt{K+1}} {\bf \Phi_{R}} \, {\bf H_{b}} \, {\bf \Phi_{T}}^{H} \right)\, {\bf w_{T}}(k)}_{h^{eq}(k)} \, x(k) + n(k).\\
h^{eq}(k) &= {\bf w_{R}}^{H}(k) \left( \sqrt{\frac{K \, \overline{\gamma}}{K+1}} \begin{bmatrix}
       e^{j \phi^{1,1}} & \dots       & e^{j \phi^{1,M_{R}}}            \\[0.3em]
       \vdots      & \ddots      & \vdots      \\[0.3em]
       e^{j \phi^{M_{T},1}}           & \dots       & e^{j \phi^{M_{T},M_{R}}}
     \end{bmatrix} + \frac{1}{\sqrt{K+1}} \begin{bmatrix}
       \Phi_{R,1}(\varphi_{1}(k)) & \dots       & \Phi_{R,1}(\varphi_{Q}(k))            \\[0.3em]
       \vdots      & \ddots      & \vdots      \\[0.3em]
       \Phi_{R,M_{R}}(\varphi_{1}(k))           & \dots       & \Phi_{R,M_{R}}(\varphi_{Q}(k))
     \end{bmatrix} \right.
\end{align}
\[\left. \begin{bmatrix}
\beta_{1}(k) &  0  & \ldots & 0\\
0  &  \beta_{2}(k) & \ldots & 0\\
\vdots & \vdots & \ddots & \vdots\\
0  &   0       &\ldots & \beta_{Q}(k)
\end{bmatrix}  \begin{bmatrix}
       \Phi^{*}_{T,1}(\varphi_{1}(k)) & \dots       & \Phi^{*}_{T,M_{T}}(\varphi_{1}(k))            \\[0.3em]
       \vdots      & \ddots      & \vdots      \\[0.3em]
       \Phi^{*}_{T,1}(\varphi_{Q}(k))           & \dots       & \Phi^{*}_{T,M_{T}}(\varphi_{Q}(k))
     \end{bmatrix} \right) {\bf w_{T}}(k)\]
\setcounter{equation}{\value{mytempeqncnt}+2}
\hrulefill
\vspace*{4pt}
\end{figure*}

\begin{figure*}[!t]
\normalsize
\setcounter{mytempeqncnt}{\value{equation}}
\setcounter{equation}{20}
\begin{align}
\label{38}
r_{p}(k) &= h_{p}(k) \, x_{p}(k) + \underbrace{{\bf \tilde{w}_{R}}^{H}(k) \left(\sqrt{\frac{K_{sp}}{K_{sp}+1}} {\bf \overline{H}_{bs}^{sp}} + \frac{1}{\sqrt{K_{sp}+1}} {\bf \tilde{\Phi}_{R}} \, {\bf H_{b}^{sp}} \, {\bf \Phi_{T}}^{H} \right)\, {\bf w_{T}}(k)}_{h_{sp}^{eq}(k)} \, x_{s}(k) + n_{p}(k). \\
r_{s}(k) &= \underbrace{{\bf w_{R}}^{H}(k) \left(\sqrt{\frac{K_{ps}}{K_{ps}+1}} {\bf \overline{H}_{bs}^{ps}} + \frac{1}{\sqrt{K_{ps}+1}} {\bf \Phi_{R}} \, {\bf H_{b}^{ps}} \, {\bf \tilde{\Phi}_{T}}^{H} \right)\, {\bf \tilde{w}_{T}}(k)}_{h_{ps}^{eq}(k)} \, x_{p}(k)
\end{align}
\[+ \underbrace{{\bf w_{R}}^{H}(k) \left(\sqrt{\frac{K_{s}}{K_{s}+1}} {\bf \overline{H}_{bs}^{s}} + \frac{1}{\sqrt{K_{s}+1}} {\bf \Phi_{R}} \, {\bf H_{b}^{s}} \, {\bf \Phi_{T}}^{H} \right)\, {\bf w_{T}}(k)}_{h_{s}^{eq}(k)} \, x_{s}(k) + n_{s}(k). \]
\setcounter{equation}{\value{mytempeqncnt}+2}
\hrulefill
\vspace*{4pt}
\end{figure*}

\subsection{Smart basis patterns: Maintaining the Secondary Channel Reliability using Artificial Diversity}
It is important to note that the application of RAB will induce fluctuations not only in the interference channels, but in the SU-to-SU channel as well. Consequently, if the secondary channel is originally Rayleigh-faded, RAB will not alter its statistics. However, if the secondary channel is Rician (reliable LoS channel), RAB will turn it into a severely-faded one. Therefore, the techniques explained in the previous subsection are applicable for the Rician-Rayleigh scenario, but not for the Rician-Rician scenario. If the secondary channel is a reliable LoS channel, we aim at maintaining its reliability while inducing fluctuations in the interference channels. This can be achieved by a technique that we refer to as {\it artificial receive diversity}. Based on (\ref{Th1}), the SU-to-SU equivalent channel can be given by $h_{s}^{eq}(k) = U_{s}(k)+V_{s}(k)$. The scattered component $V_{s}(k)$ follows a Rayleigh distribution, while the artificial fading component can be formulated by rewriting (\ref{Th1}) as
\[U_{s}(k) =\]
\begin{equation}
\label{Bas1}
\sqrt{\frac{K_{s} \overline{\gamma}_{s}}{(K_{s}+1)M_{T}M_{R}}} \sum_{m=1}^{M_{R}} \underbrace{\left(\sum_{l=1}^{M_{T}} e^{j(\theta_{T,l}(k)+\phi_{s}^{l,m})}\right)}_{a_{m}} e^{j\theta_{R,m}(k)},
\end{equation}
where $a_{m}$ is the superposition of the responses of the $M_{T}$ transmit basis patterns towards the $m^{th}$ receive basis pattern. At high SNR, it is required to maintain the reliability of the SU-to-SU channel, which can be done by applying Maximum Ration Combining (MRC) at the receive basis patterns. This can be simply applied by setting $\theta_{R,m}(k) = -\angle a_{m}$. Note that $a_{m}$ does not need to be estimated at the SU receiver as the values of $\theta_{T,l}(k)$ are pseudo-randomly generated by the SU transmitter, and the values of $\phi_{s}^{l,m}$ are deterministic and depend on the antenna structure. The artificial fading component in this case will follow a Chi-Square distribution with $M_{R}$ degrees of freedom, i.e., $U_{s}(k) \sim \mathcal{X}(2M_{R})$. Thus, the more basis patterns at the receive ESPAR antenna, the less fluctuations observed in $U_{s}(k)$. In this case, the receive basis patterns act as smart basis patterns, and because this diversity scheme averages the fluctuations that were artificially induced by the SU transmitter, we term it as artificial diversity. For an asymptotically large $M_{R}$, the LoS component can be perfectly reconstructed.

It is worth mentioning that, although the receive basis patterns act as smart antennas for the SU-to-SU link, they still act as dumb antennas for the PU-to-SU link, as the ESPAR weights at the SU receiver are not selected based on the PU-to-SU CSI and fluctuations are still induced in the interference channels. An interesting analogy between the artificial diversity and spread-spectrum systems can be constructed. In spread-spectrum systems, the transmitter {\it spreads the spectrum} of the transmitted signal using a {\it random code sequence} to minimize the ability of a hostile transceiver to {\it intercept/Jam} such signal, and the receiver {\it de-spreads} the spectrum of the signal before decoding. In the proposed artificial diversity scheme, the SU transmitter {\it spreads the pdf} of the transmitted signal over a larger dynamic range using {\it random weights} for orthogonal basis patterns to minimize the {\it PU-to-SU} and the {\it SU-to-PU} interference, and the SU receiver {\it de-spreads} the pdf of the signal before decoding. 

\section{The Multiple Exposed Secondary Users Problem: Opportunistic Scheduling using RAB}

In this section, we shift from {\it temporal} opportunistic power allocation in point-to-point channels, to {\it spatial} opportunistic scheduling in multiuser channels. In a conventional multiuser network, selecting the user with the best instantaneous channel condition results in a Multiuser Diversity (MD) gain in the achievable sum capacity. This scheme is known as Dynamic Time Division Multiple Access (D-TDMA). In addition to the conventional MD gain, Zhang and Liang introduced the concept of Multiuser Interference Diversity (MID) in \cite{25}. MID can be exploited to improve the sum capacity gain in underlay cognitive schemes by scheduling the user with the maximum SINR at each time slot. Thus, MID involves exploiting the spatial diversity in the SU-to-SU, SU-to-PU, and PU-to-SU channels. In \cite{28}, it was shown that the sum capacity of underlay cognitive radio scales like $\log(N)$ when only interference constraints at the PU receiver are imposed. This means that the underlay cognitive capacity grows faster than the conventional multiuser channel capacity, which grows like $\log(\log(N))$ \cite{6}.

In the seminal work in \cite{6}, it was shown that LoS channels hinder the achievable MD gain. It is expected that a similar effect would be encountered when considering LoS interference in an underlay cognitive setting. In this section, we aim at answering the following question: ``{\it How would LoS interference affect the SU capacity scaling?}" This section can be considered as a generalization for the {\it exposed SU problem} when multiple SUs exist. While there exists multiple settings for the multiple SU scenario, we limit our discussion to the {\it Parallel Access Channel} (PAC) for the following reasons. First, unlike the multi-access and broadcast channels, the PAC involves the impact of both multiuser transmit and receive interference diversity \cite{25}. Second, this setting is applicable for multiuser underlay D2D communications in cellular systems \cite{31}.

\begin{figure}[!t]
\centering
\includegraphics[width=2.5 in]{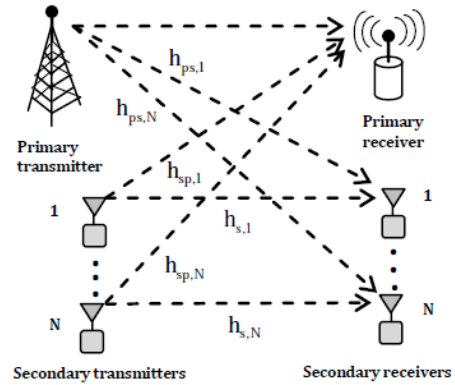}
\caption{An $N$ user parallel access cognitive channel.}
\end{figure}

\subsection{Multiuser PAC Problem Setting}
Consider an extension for the system model presented in Section II, where $N$ independent SU transmitter-receiver pairs coexist with a PU transmitter-receiver pair. Fig. 6 depicts the $N$-user underlay cognitive scheme, where $(h_{s,1}, h_{s,2},...,h_{s,N})$, $(h_{sp,1}, h_{sp,2},...,h_{sp,N})$, and $(h_{ps,1}, h_{ps,2},...,h_{ps,N})$ denote the sets of $N$ independent and identically distributed SU-to-SU, SU-to-PU, and PU-to-SU channels, respectively. Assume time-slotted transmission, where the fading channels change independently every slot. Because we are only interested in obtaining the capacity scaling laws, we relax the average interference constraint and consider a peak constraint of $Q_{p}$ only. Among all SU pairs, only the one with the best SINR is scheduled at each time slot. Thus, the selected SU pair $n^{*}$ at a certain time slot is
\begin{equation}
\label{sch}
n^{*} = \arg \, \max_{n \in \{1,2,...,N\}} \frac{\gamma_{s,n} \frac{Q_{p}}{\gamma_{sp,n}}}{1 + \overline{\gamma}_{p} \gamma_{ps,n}}.
\end{equation}
Now assume a reference network, which does not include a PU link. In such network, which corresponds to a conventional PAC, the selected user at a certain time slot is $n^{*} = \arg \, \max_{n \in \{1,2,...,N\}} \gamma_{s,n}$. We define the MD gain $\overline{\gamma}(N)$ as the ratio between the selected SU average channel power in an $N$ SUs network to the single SU average channel power. Therefore, the MD gain for the reference network is given by $\overline{\gamma}_{o}(N) = \frac{\mathbb{E}\{\gamma_{s,n^{*}}\}}{\mathbb{E}\{\gamma_{s}\}}$. Assuming Rayleigh fading for all SU links, $\mathbb{E}\{\gamma_{s,n^{*}}\} = \overline{\gamma}_{s} \, H_{N}$ \cite{19}, where $H_{N}$ is the $N^{th}$ harmonic number defined as $H_{N} = 1 + \frac{1}{2} + \frac{1}{3} + ...+ \frac{1}{N}$. For large number of users, we have $\lim_{N \to \infty} \overline{\gamma}_{o}(N)=\lim_{N \to \infty} H_{N} = \log(N)-\gamma$, where $\gamma$ is the Euler-Mascheroni constant. Thus, $\lim_{N \to \infty} \overline{\gamma}_{o}(N) \approx \log(N)$. The ergodic capacity of the reference network is given by $\mathbb{E}\left\{\log(1+\gamma_{s,n^{*}})\right\}$, and is upper-bounded by $\log(1+\mathbb{E}\left\{\gamma_{s,n^{*}}\right\})$ due to the Jensen's Inequality and the concavity of the logarithmic function. Therefore, the sum capacity of the reference network, which represents a conventional multiuser scheme, grows like $\log(\log(N))$. In the next subsection, we obtain the scaling laws for the underlay cognitive PAC scheme, and study the impact of LoS interference on the capacity growth rate.

\subsection{Capacity Scaling for Exposed SUs}
Let $\overline{\gamma}_{PAC}(N)$ be the ratio between the selected SU pair average SINR in an $N$ SUs network to the single SU pair average SINR. This quantity jointly describes the MD and MID gains. The SU pair is selected as shown in (\ref{sch}). A lower bound on $\overline{\gamma}_{PAC}(N)$ is given by
\begin{align}
\label{sch2}
\overline{\gamma}_{PAC}(N) &\stackrel{(a)}{\geq} \frac{\mathbb{E}\left\{\frac{ \gamma_{s,n^{'}} Q_{p}}{\gamma_{sp,n^{'}}({1 + \overline{\gamma}_{p} \gamma_{ps,n^{'}}})} \right\}}{\mathbb{E}\left\{\frac{ \gamma_{s,n} Q_{p}}{\gamma_{sp,n}({1 + \overline{\gamma}_{p} \gamma_{ps,n}})} \right\}} \\
&\stackrel{(b)}{=} \frac{\mathbb{E}\left\{\gamma_{s,n^{'}}\right\}}{\mathbb{E}\left\{\gamma_{s,n}\right\}} \frac{\mathbb{E}\left\{\frac{Q_{p}}{\gamma_{sp,n^{'}}({1 + \overline{\gamma}_{p} \gamma_{ps,n^{'}}})} \right\}}{\mathbb{E}\left\{\frac{Q_{p}}{\gamma_{sp,n}({1 + \overline{\gamma}_{p} \gamma_{ps,n}})} \right\}} \\
&\stackrel{(c)}{=} \overline{\gamma}_{o}(N),
\end{align}
where the SU pair $n^{'}$ is the one with the largest SU-to-SU channel. The inequality (a) is due to the fact that selecting the SU pair with the largest $\gamma_{s,n}$ is suboptimal (because the user with the maximum $\gamma_{s,n}$ is not necessarily the one with the largest SINR), (b) follows from the independence of the SU-to-SU and the interference channels, and (c) is due to the fact that $\mathbb{E}\left\{\frac{Q_{p}}{\gamma_{sp,n^{'}}({1 + \overline{\gamma}_{p} \gamma_{ps,n^{'}}})} \right\} = \mathbb{E}\left\{\frac{Q_{p}}{\gamma_{sp,n}({1 + \overline{\gamma}_{p} \gamma_{ps,n}})} \right\}$. This inequality implies that multiuser diversity offers larger gain in the underlay cognitive scheme compared to conventional multiuser networks. This follows from the fact that we can exploit MID as well as conventional MD. While $\overline{\gamma}_{o}(N)$ acts as a lower bound for $\overline{\gamma}_{PAC}(N)$, we need to figure out the growth rate of $\overline{\gamma}_{PAC}(N)$ with the number of users $N$. This can be decided by the following upper-bound
\[\overline{\gamma}_{PAC}(N) \leq\]
\begin{align}
\label{sch3}
\frac{\mathbb{E}\left\{ \max_{n} \gamma_{s,n} \right\} \mathbb{E}\left\{ \max_{n} \frac{Q_{p}}{(1 + \overline{\gamma}_{p} \gamma_{ps,n})} \right\} \mathbb{E}\left\{\max_{n} \frac{1}{\gamma_{sp,n}} \right\}}{\mathbb{E}\left\{\gamma_{s,n} \right\} \mathbb{E}\left\{ \frac{Q_{p}}{(1 + \overline{\gamma}_{p} \gamma_{ps,n})} \right\} \mathbb{E}\left\{ \frac{1}{\gamma_{sp,n}} \right\}}.
\end{align}
The inequality in (\ref{sch3}) is due to the fact that the SU pair with the best SU-to-SU channel does not necessarily have the minimum PU-to-SU and SU-to-PU channels. In the following theorem, we construct an upper-bound on the growth rate of $\overline{\gamma}_{PAC}(N)$.

\begin{thm}
{\it When the $N$ SU pairs are exposed to the PU terminals via LoS links with $K$-factors $K_{sp}$ and $K_{ps}$, and all SU-to-SU channels include a LoS component with a $K$-factor of $K_{s}$, the combined MD-MID gain $\overline{\gamma}_{PAC}(N)$ growth rate is upper-bounded by }
\[
\frac{\left( \sqrt{\frac{\log(N)}{K_{s}+1}} +  \sqrt{\frac{K_{s}}{K_{s}+1}}\right)^{2} + \mathcal{O}(\log(\log(N)))}{\left( \sqrt{\frac{1}{N(K_{sp}+1)}} +  \sqrt{\frac{K_{sp}}{K_{sp}+1}}\right)^{2} + \mathcal{O}(\log(N))} \times \]
\[
\frac{1}{\left( \sqrt{\frac{1}{N(K_{ps}+1)}} +  \sqrt{\frac{K_{ps}}{K_{ps}+1}}\right)^{2} + \mathcal{O}(\log(N))}.\]
\end{thm}

\begin{proof} See Appendix C. \IEEEQEDhere
\end{proof}
In Theorem 2, we assume that $\overline{\gamma}_{p} >> 1$ and that all SU pairs have the same average SNR. A system with non-identical average SNR for SU pairs can achieve the same growth rate in Theorem 2 using a proportional fair algorithm \cite{6}. In the following subsections, we describe the capacity scaling for the scenarios examined before in Section III.

\subsubsection{The Rician-Rician Scenario}
For asymptotically large $K$-factors, the growth rate in theorem 2 tends to be a constant term. This is due to the fact that large value of $K_{s}$ hinders the MD gain, and large values of $K_{sp}$ and $K_{ps}$ hinder the transmit and receive MID. The impact of the SU-to-SU LoS transmission is a reduction of $\frac{1}{K_{s}+1}$ in the growth term of $\log(N)$. On the other hand, the impact of SU-to-PU LoS interference is a reduction of $\frac{1}{K_{sp}+1}$ in the $\frac{1}{N}$ factor, thus the term $\sqrt{\frac{1}{N(K_{sp}+1)}} + \sqrt{\frac{K_{sp}}{K_{sp}+1}} \to 1$ when $K_{sp} \to \infty$.

\subsubsection{The Rician-Rayleigh Scenario}
For asymptotically large $K_{sp}$ and $K_{ps}$, the terms $\sqrt{\frac{1}{N(K_{sp}+1)}} + \sqrt{\frac{K_{sp}}{K_{sp}+1}} \to 1$ and $\sqrt{\frac{1}{N(K_{ps}+1)}} + \sqrt{\frac{K_{ps}}{K_{ps}+1}} \to 1$. Contrarily, the SU-to-SU channel is subject to Rayleigh fading, so $K_{s}$ = 0. Thus, $\overline{\gamma}_{PAC}(N)$ approximately grows like $\log(N)$, and the capacity grows like $\log(\log(N))$. This is the same growth rate of the conventional (reference) multiuser network with no MID. Therefore, LoS interference hinders MID gain and the capacity follows the same scaling law of the reference network.

\subsubsection{The Rayleigh-Rayleigh Scenario}
All $K$-factors are set to zero, and the growth rate in Theorem 2 reduces to $N^{2} \log(N)$. Thus, the SU capacity scales like $\log(N)+\mathcal{O}(\log(\log(N)))$, which means that the multiuser gain scales faster with the number of users in the underlay cognitive system compared to conventional interference free system. This is because the SUs exploit MID as well as the MD offered by the SU-to-SU channel fluctuations.

\subsubsection{The Rayleigh-Rician Scenario}
The growth rate of $\overline{\gamma}_{PAC}(N)$ becomes $N^{2}\left( \sqrt{\frac{\log(N)}{K_{s}+1}} +  \sqrt{\frac{K_{s}}{K_{s}+1}}\right)^{2}$. Consequently, the capacity still scales like $\log(N)+\mathcal{O}(\log(\log(N)))$. This implies that the impact of LoS SU-to-SU channel is not as substantial as the impact of LoS interference, which changes the growth rate of SU capacity scaling by suppressing the $\log(N)$ term. These results lead to the following remark.\\
{\bf Remark 4:} The sum capacity of multiuser underlay cognitive radio systems scales like $\log(N)$, while the sum capacity of a conventional multiuser system scales like $\log(\log(N))$. Thus, multiuser diversity gain is larger in the underlay cognitive network due to the exploitation of MID in addition to MD. However, when the SUs are exposed to LoS mutual interference, the multiuser diversity characteristics are altered and the scaling law reduces to $\log(\log(N))$ for asymptotically large $K$-factors. We denote this problem as the {\it multiple exposed SUs problem}. \\

An intuitive approach to restore the MID gain hindered by LoS interference is to induce fluctuations in the interference channels using RAB. Because mobile SU transmitters can not use multiple dumb antennas to randomize the channel due to space limitations, employing RAB using a single radio ESPAR antenna is a perfect solution to the multiple exposed SUs problem. This is investigated in the following subsection.

\subsection{Opportunistic Scheduling using RAB}
We investigate the capacity scaling when the RAB technique is employed at both SU terminals. Assume that for each SU pair, the SU transmitter and receiver applies RAB with $M_{T}$ and $M_{R}$ basis patterns, respectively . In this case, the selected SU pair $n^{*}$ at each time slot is
\begin{equation}
\label{sch5}
n^{*} = \arg \, \max_{n \in \{1,2,...,N\}} \frac{\gamma^{eq}_{s,n} \frac{Q_{p}}{\gamma^{eq}_{sp,n}}}{1 + \overline{\gamma}_{p} \gamma^{eq}_{ps,n}},
\end{equation}
where $\gamma^{eq}_{s,n}$, $\gamma^{eq}_{sp,n}$, and $\gamma^{eq}_{ps,n}$ are the equivalent SU-to-SU, SU-to-PU, and PU-to-SU equivalent channel powers after applying RAB. The signal model for the equivalent channels is given in (\ref{Th1}). The key idea behind the application of RAB in the multiuser setting is to induce artificial fluctuations in the interference channels. While classical opportunistic beamforming introduced in \cite{6} intended to schedule users at the peaks of their channel gains, the application of RAB in the underlay cognitive system intends to schedule users at the nulls of the interference channels. Thus, RAB performs {\it opportunistic nulling} and eliminates the impact of LoS interference emphasized in remark 4. In the following theorem, we derive the scaling law for opportunistic scheduling using RAB.

\begin{thm}
{\it When the $N$ SU pairs are exposed to the PU terminals via LoS links, and the SU transmitters and receivers apply RAB with $M_{T}$ and $M_{R}$ basis patterns, respectively, the combined MD-MID gain $\overline{\gamma}_{PAC}(N)$ growth rate is upper-bounded by }
\[
N^{2}\left(\left( \sqrt{\frac{\log(N)}{K_{s}+1}} +  \sqrt{\frac{M_{T}M_{R}K_{s}}{K_{s}+1}}\right)^{2} + \mathcal{O}(\log(\log(N)))\right).\]
\end{thm}
\begin{proof} See Appendix D. \IEEEQEDhere
\end{proof}
Theorem 3 suggests that increasing $M_{T}$ and $M_{R}$ leads to an increase in the effective magnitude of the fixed component by a factor of $\sqrt{M_{T}M_{R}}$, but has no effect on the $\log(N)$ term in the growth rate. This can be interpreted by the impact of the number of basis patterns on the dynamic range of the SU-to-SU channel fluctautions. Increasing $M_{T}$ and $M_{R}$ leads to increasing the dynamic range of $\gamma^{eq}_{s,n}$ by increasing the maximum possible value of the SU-to-SU artificial fading component (see Appendix D). This directly leads to an improvement in the MD gain. However, the impact of MD on the overall capacity scaling is not as cruical as the impact of MID. As stated in remark 4, LoS interference causes the capacity to scale with $\log(\log(N))$ instead of $\log(N)$. From theorem 3, it is obvious that after applying RAB, the capacity scales with $\log(N)$ regardless of the $K$-factor of the interference channels.

While it is obvious from theorem 3 that increasing $M_{T}$ and $M_{R}$ improves the MD gain by increasing the SU-to-SU channel dynamic range, the impact of $M_{T}$ and $M_{R}$ on the MID gain is not straightforward. In fact, the MID gain does not depend on the dynamic range of the interference channels. Instead, it depends on the ability of the SU transmitter/receiver to null the interference to/from the PU pair. To study the impact of the number of transmit and receive basis patterns, we focus on the SU-to-PU channel, and assume two RAB schemes: one with $M_{T}$ = 2, and the other with $M_{T} \to \infty$. For a fixed and infinitesimally small $\delta$, $M_{T}$ = 2, the term $\frac{\mathbb{E}\left\{\max_{n} \frac{1}{\gamma^{eq}_{sp,n}} \right\}}{\mathbb{E}\left\{\frac{1}{\gamma^{eq}_{sp,n}} \right\}}$ in (\ref{sch3}) grows like $\frac{1}{\left( \sqrt{\frac{1}{\epsilon_{2} N(K_{sp}+1)}} +  \delta \right)^{2}}$ (see Appendix D for a detailed proof), while when $M_{T} \to \infty$, the growth rate is $\frac{1}{\left( \sqrt{\frac{1}{\epsilon_{M_{T}} N(K_{sp}+1)}} +  \delta \right)^{2}}$, where $\epsilon_{2}$ and $\epsilon_{M_{T}}$ are the fractions of users that have an interference channel gain less than $\delta$. If the numbers of users meeting this criterion follow a binomial distribution with mean values of $\epsilon_{2} N$ and $\epsilon_{M_{T}} N$, and assuming a large $K$-factor then
\[\epsilon_{2} = P\left(\sqrt{\frac{K_{sp} \overline{\gamma}_{sp}}{2(K_{sp}+1)}} \, \left|\sum_{l=1}^{2} e^{j(\theta_{T,l}(k)+\phi_{sp}^{l})}\right| < \delta \right),\]
and
\[\epsilon_{M_{T}} = P\left(\sqrt{\frac{K_{sp} \overline{\gamma}_{sp}}{M_{T}(K_{sp}+1)}} \, \left|\sum_{l=1}^{M_{T}} e^{j(\theta_{T,l}(k)+\phi_{sp}^{l})}\right| < \delta\right).\]
From theorem 1, we know that when $M_{T} \to \infty$, then $\sqrt{\frac{K_{sp} \overline{\gamma}_{sp}}{M_{T}(K_{sp}+1)}} \, \left|\sum_{l=1}^{M_{T}} e^{j(\theta_{T,l}(k)+\phi_{sp}^{l})}\right|$ will follow a Rayleigh distribution. This implies that $\epsilon_{M_{T}}  = 1-\mbox{exp}\left(\frac{-\delta^{2} (K_{sp}+1)}{K_{sp} \overline{\gamma}_{sp}}\right)$. For $\delta \to 0$, and using a Taylor approximation, it can be shown that $\epsilon_{M_{T}}  \approx \frac{\delta^{2} (K_{sp}+1)}{K_{sp} \overline{\gamma}_{sp}}$. In order to derive $\epsilon_{2}$, we use the Euler identity as follows
\[\epsilon_{2} = P\left(A \left(\sum_{l=1}^{2} \mbox{cos}(\theta_{T,l}(k)+\phi_{sp}^{l})\right)^{2} \right. + \]
\[\left. A\left(\sum_{l=1}^{2}\mbox{sin}(\theta_{T,l}(k)+\phi_{sp}^{l})\right)^{2} < \delta^{2}\right),\]
where $A = \frac{K_{sp} \overline{\gamma}_{sp}}{2(K_{sp}+1)}$. In order to obtain $\epsilon_{2}$, we need to derive the pdf of the random variable
\[\left(\mbox{cos}(\theta_{T,1}(k)+\phi_{sp}^{1})+\mbox{cos}(\theta_{T,2}(k)+\phi_{sp}^{2})\right)^{2}\]
\[ + \left(\mbox{sin}(\theta_{T,1}(k)+\phi_{sp}^{1})+ \mbox{sin}(\theta_{T,2}(k)+\phi_{sp}^{2})\right)^{2},\]
which can be rewritten as
\[2 + 2 \, \mbox{cos}(\theta_{T,1}(k)+\phi_{sp}^{1}-\theta_{T,2}(k)-\phi_{sp}^{2}).\]
It can be easily shown that $\theta_{T,1}(k)+\phi_{sp}^{1}-\theta_{T,2}(k)-\phi_{sp}^{2}$ follow a uniform distribution, thus $1 + \mbox{cos}(\theta_{T,1}(k)+\phi_{sp}^{1}-\theta_{T,2}(k)-\phi_{sp}^{2})$ follows the distribution in (\ref{pdfeq}). Let $y = \frac{K_{sp} \overline{\gamma}_{sp}}{(K_{sp}+1)} (1 + \mbox{cos}(\theta_{T,1}(k)+\phi_{sp}^{1}-\theta_{T,2}(k)-\phi_{sp}^{2}))$. Hence, $\epsilon_{2}$ can be evaluated as
\begin{align}
\label{eqs1}
\epsilon_{2} &= \frac{(K_{sp}+1)}{K_{sp} \overline{\gamma}_{sp}} \int_{y=0}^{\delta^{2}} \frac{1}{\pi \sqrt{1-(1-y)^{2}}} dy \\
&= \frac{(K_{sp}+1)}{K_{sp} \overline{\gamma}_{sp}} \left( \frac{1}{\pi}\mbox{sin}^{-1}(\delta^{2}-1) + \frac{1}{2} \right).
\end{align}
It can be shown that $\left(\frac{1}{\pi} \mbox{sin}^{-1}(\delta^{2}-1) + \frac{1}{2} \right) > \delta^{2}, \forall \delta > 0$, which means that $\epsilon_{2} = \frac{(K_{sp}+1)}{K_{sp} \overline{\gamma}_{sp}} \left( \frac{1}{\pi}\mbox{sin}^{-1}(\delta^{2}-1) + \frac{1}{2} \right) >$  $\epsilon_{M_{T}} = \frac{(K_{sp}+1) \delta^{2}}{K_{sp} \overline{\gamma}_{sp}}$. Because $\overline{\gamma}_{PAC}(N)$ is proportional to $\epsilon_{M_{T}} N$, the MID gain is larger in the case when $M_{T}$ = 2 than in the case when $M_{T} \to \infty$. It is important to note that $M_{T} \to \infty$ corresponds to the case of Rayleigh fading. Thus, in a LoS interference environment, when only two dumb basis patterns (1 parasitic element) is employed at the SU transmitter, the MID gain is better than that in a Rayleigh fading interference channel. This is because MID does not depend on the dynamic range of the interference channel fluctuations, but depends on how frequent this interference channel is set to an arbitrarily small gain $\delta$. In Fig. 7, we plot the pdf of the artificial fading component after applying RAB with various numbers of basis patterns. It is clear that with 2 basis patterns, the artificial fading component is highly likely to either have the maximum value of $\sqrt{2}$, or have a small value near zero. As the number of basis patterns increase, the dynamic range increases and the pdf converges to an exponential distribution. It is clear that for an an arbitrarily small gain $\delta$, the probability that the artificial fading channel becomes less than $\delta$ is highest for the case of 2 basis patterns. Thus, two basis patterns null the interference more frequently than the conventional Rayleigh fading scenario. This leads to the following remark.\\
{\bf Remark 5:} LoS interference hinders the achievable MID gain and changes the capacity scaling law from $\log(N)$ to $\log(\log(N))$. Using dumb basis patterns, not only can one restore the capacity scaling law of $\log(N)$, but LoS interference can be exploited as well. Using 2 basis patterns, the MID can be improved compared to the Rayleigh fading scenario due to the desirable characteristics of the artificial fading channel pdf.

The conclusions in remark 5 have a great impact on the hardware implementation of RAB using ESPAR antennas. Because we only need one parasitic antenna, the value of the reactive load needed to adjust the random weights of the basis patterns can be obtained in closed form (refer to the ESPAR antenna model in Section II). Moreover, an ESPAR with one parasitic antenna entails less complexity and more compactness.

\begin{figure}[!t]
\centering
\includegraphics[width=3.5 in]{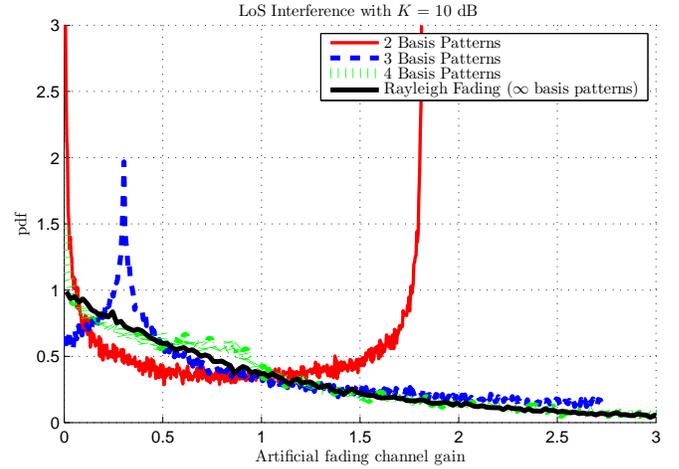}
\caption{Probability density function of the artificial fading component for various numbers of basis patterns.}
\end{figure}

\section{Numerical results}
This section provides numerical results for the techniques presented throughout the paper. Monte-Carlo simulations are carried out and results are averaged over 100,000 runs. We assume the following parameter settings: $N_{o} = 1$ W/Hz, $\overline{\gamma}_{sp}=\overline{\gamma}_{ps} = \overline{\gamma}_{s} = 0$ dB, and $\overline{\gamma}_{p} = 10$ dB. For all Rician channels, we assume a $K$-factor of 10 dB unless otherwise stated. Fig. 8 depicts the ergodic capacity of the SU as a function of the average interference power constraint $Q_{av}$. We define the factor $\rho = \frac{Q_{p}}{Q_{av}}$ and plot the ergodic capacity for $\rho = \infty$ (no peak interference constraint) and $\rho = 1.2$. As expected, the SU capacity is a monotonically increasing function of $Q_{av}$. In other words, the SU is allowed to transmit with higher power when the interference constraint is relaxed. We can also observe that when a joint peak and average interference constraint is imposed, the capacity decreases, and the amount of degradation is more significant for smaller values of $Q_{av}$.

It is notable that for all fading scenarios, the SU capacity is larger than the AWGN capacity, which agrees with the results in \cite{2}. The AWGN capacity is a special case of the Rician-Rician scenario when $K_{sp} = K_{ps} = K_{s} = \infty$. Thus, the AWGN channel is an extreme case of the LoS interference scenario tackled in Section III. Note that at $\overline{\gamma}_{s}=0$ dB, the Rayleigh-Rician scenario offers the best SU capacity because it enjoys a reliable SU-to-SU link, and severely-faded interference channels, which matches with the conclusions in remark 2. LoS interference can significantly degrade the SU capacity, as a capacity gap of 1.05 bps/Hz is observed between the Rayleigh-Rician and Rician-Rician scenarios. Also, the Rayleigh-Rayleigh scenario offers a capacity increase of about 0.75 bps/Hz more than the Rician-Rician scenario, which means that the degradation caused by severely faded SU-to-SU link is less harmful than LoS interference. The Rician-Rayleigh scenario performs worse than the Rician-Rician scenario as it suffers from both severely faded SU-to-SU channel and LoS interference. It can be also observed that all fading scenarios that include a Rayleigh faded SU-to-SU channel are more sensitive to the peak interference constraint, because the large dynamic range of the SU-to-SU channel implies that the allocated SU transmit power will ``{\it hit the peak constraint}" more frequently.

Fig. 9 depicts the impact of the SU-to-SU average SNR on the achievable SU capacity. Assuming a Rayleigh-Rician scenario, we plot the SU capacity normalized by the AWGN capacity. For $\overline{\gamma}_{s}= 0$ dB, the SU capacity is larger than the AWGN capacity for all values of $Q_{av}$. As $Q_{av}$ increases, the SU capacity gain due to opportunistic power allocation decreases and approaches unity. For $\overline{\gamma}_{s}= 10$ dB, the SU capacity is still larger than the AWGN capacity, but the achievable capacity gain is less than the case when $\overline{\gamma}_{s}= 0$ dB. For $\overline{\gamma}_{s}= 20$ dB, the SU capacity is less than half of the AWGN capacity, and the impact of $Q_{av}$ is negligible. This is because at high SU-to-SU average SNR, fading is a source of unreliability.

To eliminate the negative impact of LoS interference, the concept of opportunistic spectrum sharing using dumb basis patterns was proposed. Fig. 10 shows the pdf of the interference channel magnitude ($|h_{sp}|$ or $|h_{ps}|$) before and after applying RAB for various numbers of basis patterns. It is obvious from Fig. 10 that as the number of basis patterns increases, the pdf of the equivalent channel spreads, indicating a larger dynamic range of fluctuations. It is shown that 5 basis patterns (5 parasitic elements) are enough to convert a Rician channel to a Rayleigh one. Any further increase in the number of basis patterns will not result in an increase in the dynamic range of the channel. Fig. 11 depicts the magnitude of an interference channel ($|h_{sp}|$ or $|h_{ps}|$) versus time before and after applying RAB. It is clear that after applying RAB, the resultant channel will have a larger dynamic range and more occurrences of deep fades (marked with circles) than the LoS channel. Fig. 12 shows the pdf of an interference channel after applying RAB, and then after de-spreading at the SU receiver by the aid of artificial diversity.

In Fig. 13, we investigate the impact of the number of basis patterns on the achieved capacity gain. For the Rician-Rayleigh scenario, only one parasitic element is enough to achieve a significant capacity gain relative to the Rician-Rician scenario. Any further increase in the number of parasitic elements will make the capacity of the Rician-Rayleigh scenario converge to that of the Rayleigh-Rayleigh scenario. The same behavior is observed for the Rician-Rician scenario, where we apply the artificial diversity scheme to regain the reliability of the SU-to-SU channel.

We investigate the SU capacity scaling behavior in Fig. 14 by plotting the SU sum capacity normalized by a single SU capacity versus the number of SU pairs in a PAC. We note that the sum capacity in the Rayleigh-Rayleigh scenario grows faster than a reference network experiencing a $\log(\log(N))$ growth rate. This is due to MID, which allows the capacity to scale like $\log(N)$. As illustrated in remark 4, when LoS interference is encountered in the Rician-Rayleigh scenario, we note that the capacity scales with the same rate of the reference network. Thus, LoS interference hinders MID and the growth rate becomes $\log(\log(N))$. While the Rayleigh-Rician scenario causes a minor loss in the MD gain with no effect on the growth rate, the Rician-Rician scenario experiences no MD and MID gains, and the sum capacity is nearly constant for any number of SU pairs.

In Fig. 15, we show that using RAB with only 2 basis patterns, one can regain the MD gain losses encountered in the Rayleigh-Rician scenario. More importantly, Fig. 15 shows that LoS interference can be exploited. When applying RAB in the Rician-Rayleigh scenario, the SU sum capacity becomes superior to that of the Rayleigh-Rayleigh scenario due to the impact of opportunistic nulling, as explained in remark 5. Fig. 16 demonstrates the capacity growth rates for various scenarios by plotting the sum capacity normalized by the growth function. We see that in the Rician-Rayleigh scenario, the capacity grows like $\log(\log(N))$ as the ratio between the capacity and $\log(\log(N))$ tends to a constant as $N$ increases. This constant is known as the {\it pre-log} factor. We can also see that, while the Rayleigh-Rayleigh scenario grows like $\log(N)$, it has a less pre-log factor compared to the Rician-Rayleigh scenario with the application of RAB. Thus, applying RAB with 2 basis patterns in an environment with LoS interference improves the pre-log factor and maintains the $\log(N)$ growth rate, which interprets the results in Fig. 15.

\begin{figure}[!t]
\centering
\includegraphics[width=3.25 in]{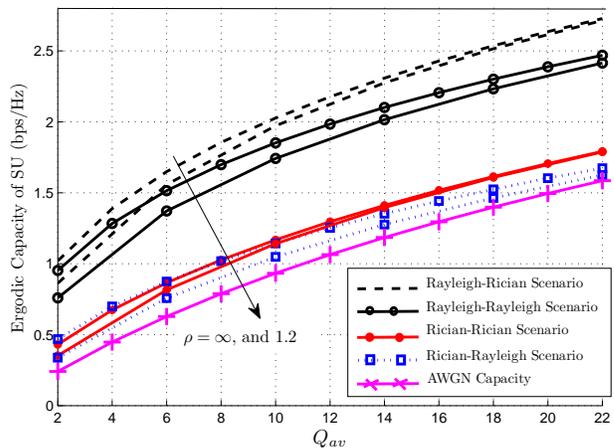}
\caption{Illustration for the impact of LoS interference on SU capacity.}
\end{figure}

\begin{figure}[!h]
\centering
\includegraphics[width=3.25 in]{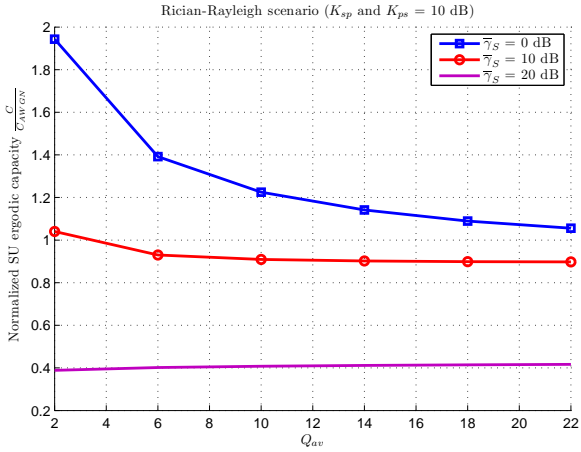}
\caption{Impact of LoS SU-to-SU channel on the SU capacity.}
\end{figure}

\begin{figure}[!h]
\centering
\includegraphics[width=3.25 in]{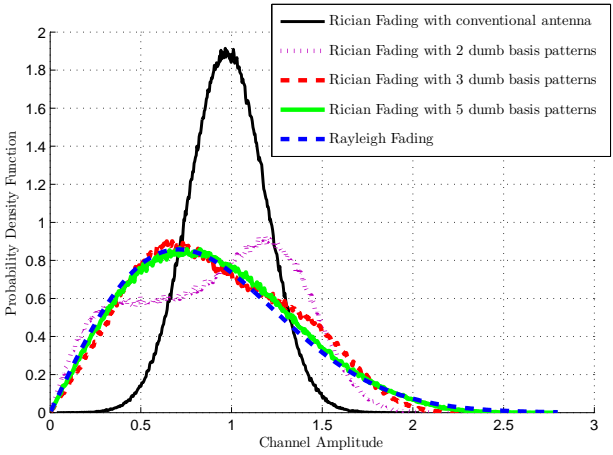}
\caption{Equivalent channel pdf after applying RAB.}
\end{figure}

\begin{figure}[!h]
\centering
\includegraphics[width=3.25 in]{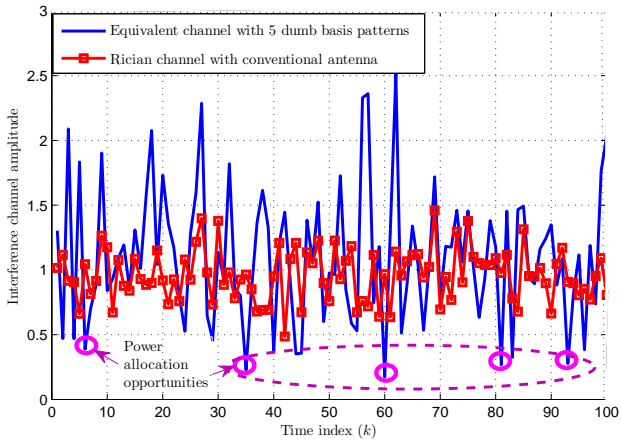}
\caption{Channel gain before and after applying RAB.}
\end{figure}

\begin{figure}[!h]
\centering
\includegraphics[width=3.25 in]{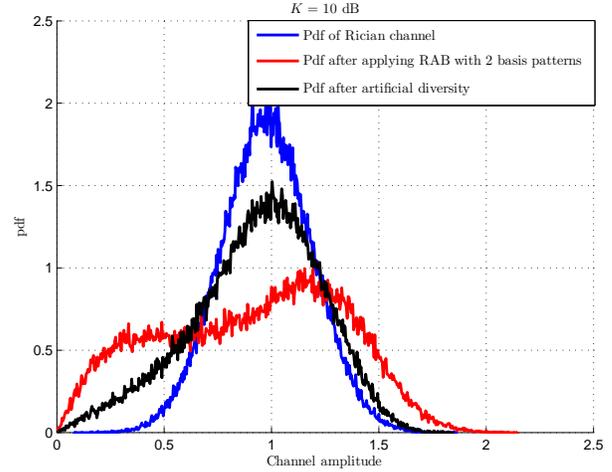}
\caption{Equivalent channel pdf after applying RAB and artificial diversity.}
\end{figure}

\begin{figure}[!h]
\centering
\includegraphics[width=3.25 in]{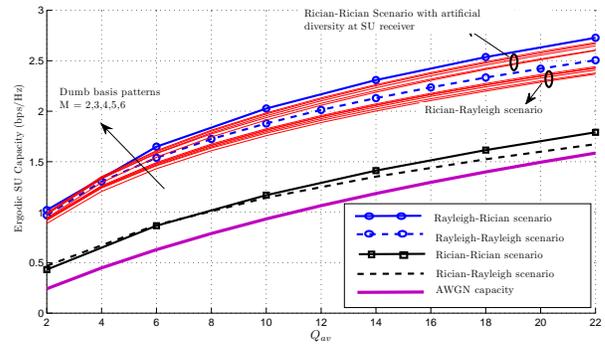}
\caption{Illustration for the impact RAB on SU capacity enhancement.}
\end{figure}

\begin{figure}[!t]
\centering
\includegraphics[width=3.5 in]{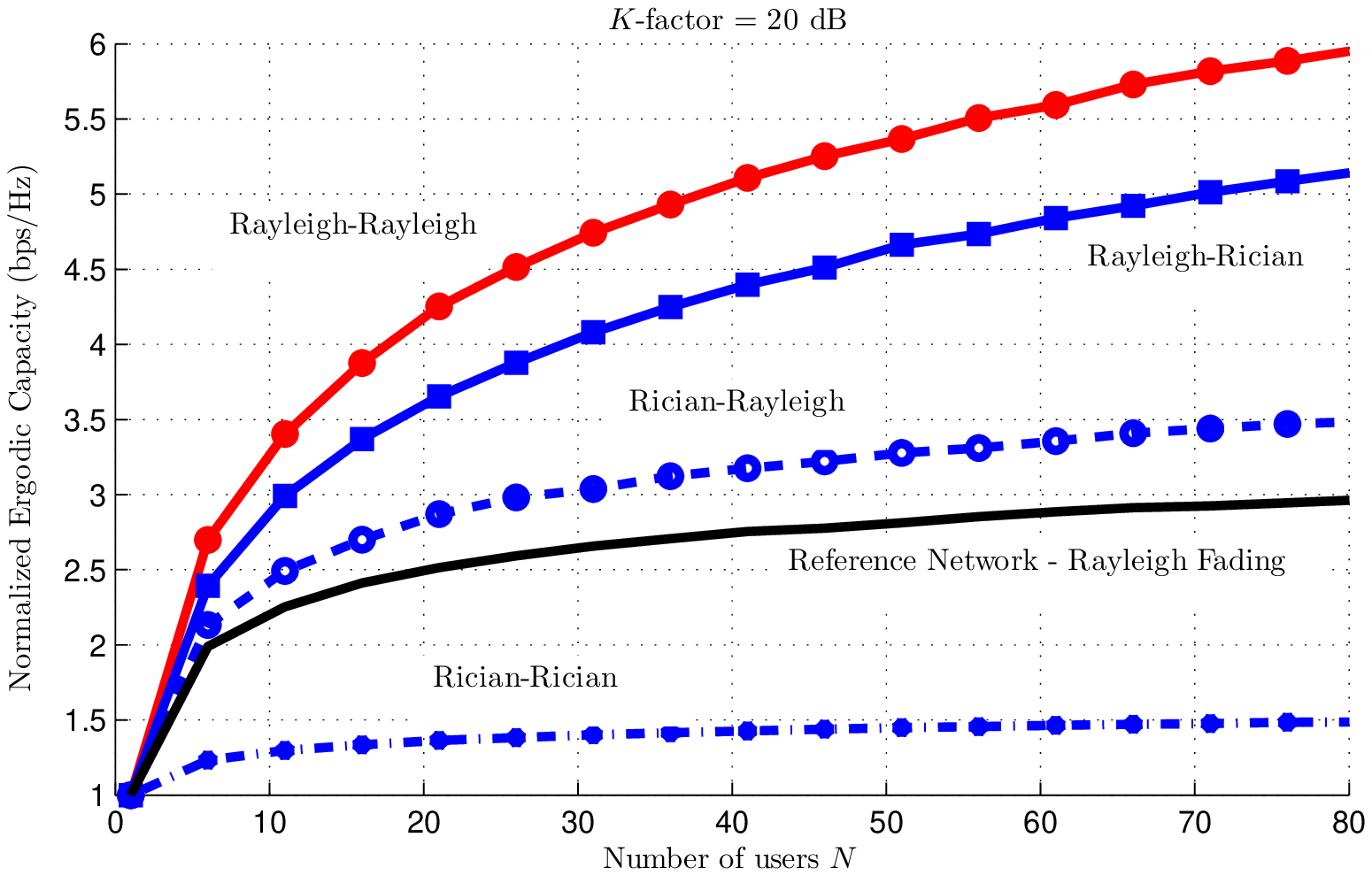}
\caption{Normalized ergodic capacity as a function of the number of SU pairs.}
\end{figure}

\begin{figure}[!t]
\centering
\includegraphics[width=3.5 in]{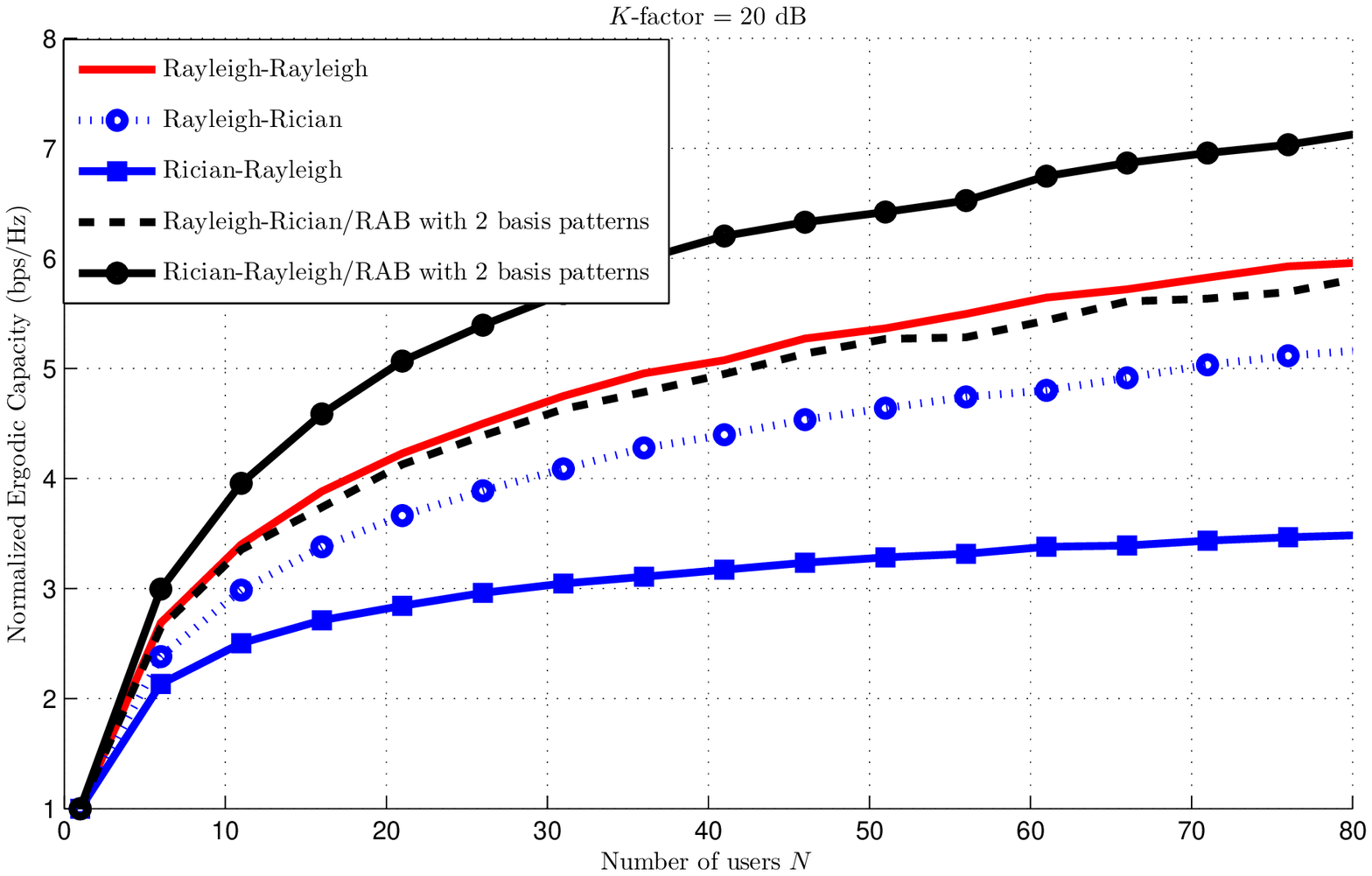}
\caption{Normalized ergodic capacity as a function of the number of SU pairs after applying RAB.}
\end{figure}

\begin{figure}[!t]
\centering
\includegraphics[width=3.5 in]{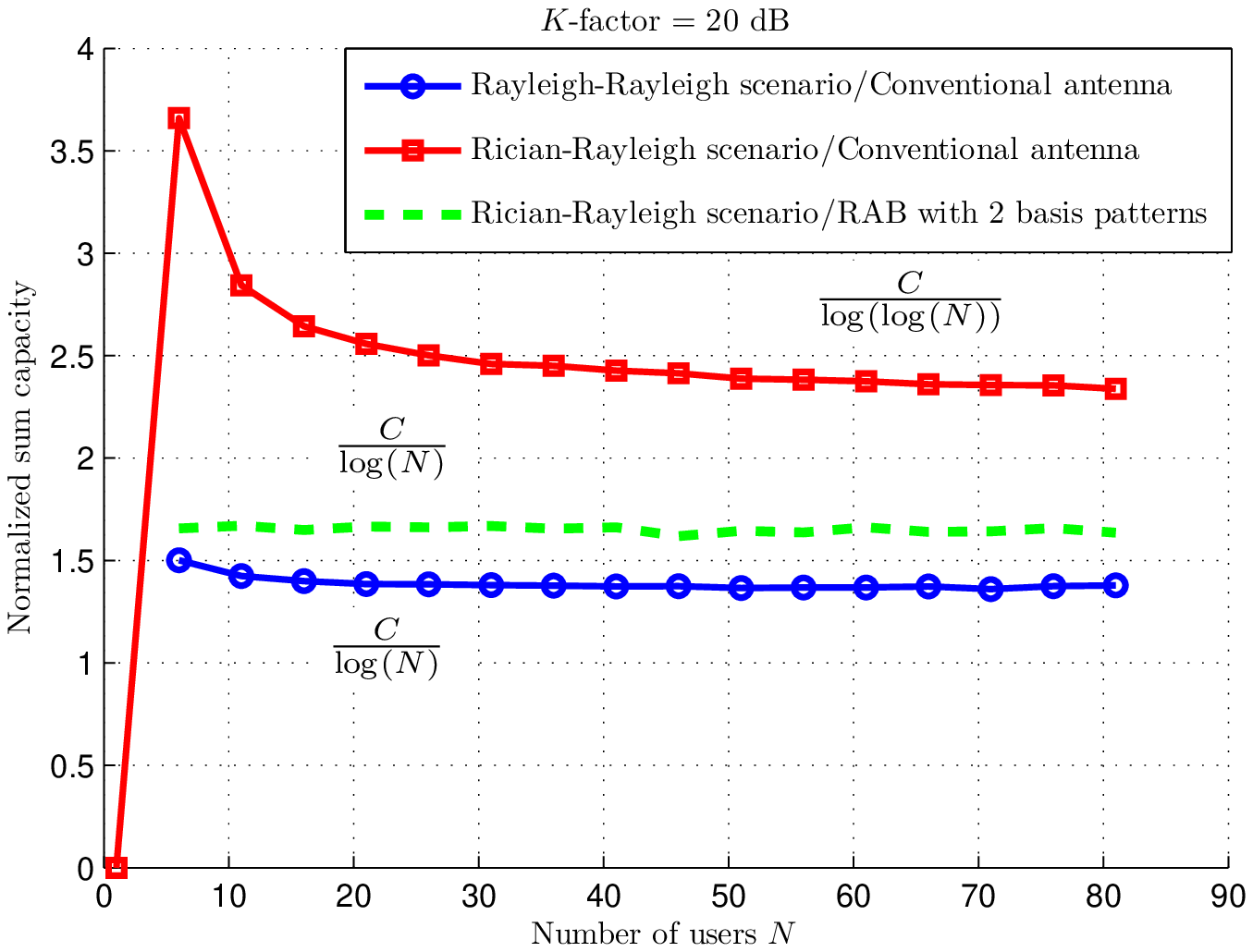}
\caption{Ergodic throughput normalized by the growth rate for several scenarios.}
\end{figure}

\section{Conclusions}
In this paper, we presented a comprehensive analysis for the negative impact of LoS mutual interference on the SU capacity in a spectrum sharing system. It was shown that when the dynamic range of the interference channel is small, the SU capacity is significantly decreased. Stemming from this point, we introduced the RAB technique to induce channel fluctuations in the interference channel using the {\it dumb basis patterns} of an ESPAR antenna. If the secondary channel contains a LoS component, we adopt an {\it artificial diversity scheme} to maintain its reliability while inducing fluctuations in the interference channels. Moreover, we investigated the impact of LoS interference on the multiuser interference diversity achieved in an underlay parallel access channel. While LoS interference alters the capacity scaling to be $\log(\log(N))$ instead of $\log(N)$, we show that using RAB, LoS interference can be actually exploited to improve the scaling pre-log factor in addition to the restoration of the $\log(N)$ growth rate. The proposed RAB scheme requires a single RF chain, and can fit within tight space limitations. Thus, it is adequate for low cost mobile transceivers. Our future work includes evaluating capacity scaling for the cognitive MAC channel \cite{33}\cite{34}, and deploying reconfigurable antennas to apply blind interference alignment \cite{32}.

\appendices

\section{Proof of Lemma 1}
\renewcommand{\theequation}{\thesection.\arabic{equation}}
By relaxing the peak interference constraint and setting $Q_{p} = \infty$ in (\ref{5}), and assuming that $\gamma_{ps}$ is deterministic and equal to $\overline{\gamma}_{ps}$, the SU capacity is given by
\[C = \int_{z = \lambda \log(2) (1+\overline{\gamma}_{ps}\overline{\gamma}_{p})}^{\infty} \log_{2}\left(\frac{z}{\lambda \log(2)(\overline{\gamma}_{ps}\overline{\gamma}_{p}+1)}\right) f_{z}(z) dz, \]
substituting $x = \frac{z}{\lambda \log(2)(\overline{\gamma}_{ps}\overline{\gamma}_{p}+1)}$, this integral can be reduced to $ \int_{x = 1}^{\infty} \log_{2}(x) f_{x}(x) dx$, which can be easily evaluated as
\[C = {\rm Ei} \left(-\lambda \log(2) (\gamma_{ps} \overline{\gamma}_{p}+1) \frac{\overline{\gamma}_{sp}}{\overline{\gamma}_{s}}\right).\]
For $\overline{\gamma}_{s} \to \infty$, we have $ \lambda \log(2) (\gamma_{ps} \overline{\gamma}_{p}+1) \frac{\overline{\gamma}_{sp}}{\overline{\gamma}_{s}} \to 0$. Thus, the SU capacity is given by $C \to \log_{2}\left(\frac{\overline{\gamma}_{s}}{\lambda \log(2) \overline{\gamma}_{sp} (\overline{\gamma}_{p} \gamma_{p} + 1)}\right) - \gamma$ [26, Sec. 8.214]. The value of $\lambda$ in an AWGN channel differs from that in a fading channel. Generally, $\lambda$ is given by
\[\lambda = \log(2)\left(Q_{av}+(\mathbb{E}\left\{\gamma_{ps}\overline{\gamma}_{p}+1)\right\}\mathbb{E}\left\{\frac{1}{z}\right\}\right).\]
For an AWGN channel, the random variables $z$ and $\gamma_{ps}$ are deterministic and $\lambda$ is given by
\[\lambda_{AWGN} = \log(2)\left(Q_{av}+(\overline{\gamma}_{ps}\overline{\gamma}_{p}+1)\frac{1}{\mathbb{E}\{z\}}\right).\]
Thus, by Jensen Inequality, $\mathbb{E}\{1/z\} \geq 1/\mathbb{E}\{z\}$, $\lambda_{AWGN} \leq \lambda$. This implies the following
\[C \lesssim \log_{2}\left(\frac{\overline{\gamma}_{s}}{\lambda_{AWGN} \log(2) \overline{\gamma}_{sp} (\overline{\gamma}_{p} \gamma_{p} + 1)}\right) - \gamma,\]
where $\log_{2}\left(\frac{\overline{\gamma}_{s}}{\lambda_{AWGN} \log(2) \overline{\gamma}_{sp} (\overline{\gamma}_{p} \gamma_{p} + 1)}\right)$ represents the AWGN capacity. For $\overline{\gamma}_{s} \to 0$, we have $ \lambda \log(2) (\gamma_{ps} \overline{\gamma}_{p}+1) \frac{\overline{\gamma}_{sp}}{\overline{\gamma}_{s}} \to \infty$. Thus, the SU capacity is given by $C \to \frac{\overline{\gamma}_{s}}{\lambda \log(2) \overline{\gamma}_{sp} (\overline{\gamma}_{p} \gamma_{p} + 1)}\exp\left(-\frac{\lambda \log(2) \overline{\gamma}_{sp} (\overline{\gamma}_{p} \gamma_{p} + 1)}{\overline{\gamma}_{s}}\right)$ [26, Sec. 8.215]. Noting that $\lambda_{AWGN} < \lambda$, the final result in the lemma directly follows.

\section{Proof of Theorem 1}
\renewcommand{\theequation}{\thesection.\arabic{equation}}
The proof applies for $h^{eq}_{s}(k)$, $h^{eq}_{sp}(k)$, and $h^{eq}_{ps}(k)$ and thus we adopt the generic channel model $h^{eq}(k)$ in (20). After applying RAB with $\mathbf{w_{R}}(k) = \left[ \frac{e^{j \theta_{R,1}(k)}}{\sqrt{M_{R}}} \,\, \ldots \,\, \frac{e^{j \theta_{R,M_{R}}(k)}}{\sqrt{M_{R}}}\right]$,  $\mathbf{w_{T}}(k) = \left[ \frac{e^{j \theta_{T,1}(k)}}{\sqrt{M_{T}}} \,\, \ldots \,\, \frac{e^{j \theta_{T,M_{T}}(k)}}{\sqrt{M_{T}}}\right]$ and performing the various matrix multiplications, $h^{eq}(k)$ can be obtained as
\[
h^{eq}(k) = \, \underbrace{\sqrt{\frac{K \overline{\gamma}}{(K+1)M_{T}M_{R}}} \sum_{l=1}^{M_{T}} \sum_{m=1}^{M_{R}} e^{j(\theta_{R,m}(k)+\theta_{T,l}(k)+\phi^{l,m})}}_{\mbox{{\fontsize{0.35cm}{1em}\selectfont Artificial fading}
} \,\,\, U(k)}.
\]
\begin{equation}
\label{Th1}
+ \underbrace{\sum_{l=1}^{M_{T}} \sum_{m=1}^{M_{R}} \sum_{i=1}^{Q} \Phi_{R,m}(\varphi_{i}(k)) \Phi^{*}_{T,l}(\varphi_{i}(k)) \beta_{i}(k) \, \frac{e^{j(\theta_{R,m}(k)+\theta_{T,l}(k))}}{\sqrt{M_{T} M_{R}}}}_{\mbox{{\fontsize{0.35cm}{1em}\selectfont Scattered component}
} \,\,\, V(k)}.
\end{equation}
The scattered component $V(k)$, which represents the underlying rich scattering environment, is known to follow a complex Gaussian distribution \cite{6}. The amplitude of this channel component is scaled by $\frac{1}{\sqrt{K+1}}$ as the variance of $\beta_{i}(k)$ is $\frac{\overline{\gamma}}{K+1}$, thus $V(k) \sim \mathcal{CN}(0,\frac{\overline{\gamma}}{K+1})$. The term $U(k)$, on the other hand, represents the artificial fading resulting from the constructively/destructively interfering LoS components perceived by different basis patterns with different phase shifts obtained by tuning the reactive loads. We are interested in deriving the pdf of this component. Using Euler identity, $U(k)$ can be represented as
\[U(k) = A \sum_{l=1}^{M_{T}} \sum_{m=1}^{M_{R}} \cos(\theta_{R,m}(k)+\theta_{T,l}(k)+\phi_{l,m})\]
\[+ j \sin(\theta_{R,m}(k)+\theta_{T,l}(k)+\phi_{l,m}),\]
where $A = \sqrt{\frac{K \overline{\gamma}}{M_{T} M_{R}(K+1)}}$. Let $U^{R}(k) = {\bf \mathfrak{Re}}\{U(k)\} = A  \sum_{l=1}^{M_{T}} y_{l}$, where $y_{l} = \sum_{m=1}^{M_{R}} y_{l,m}$, and $y_{l,m} = \cos(\theta_{R,m}(k)+\theta_{T,l}(k)+\phi_{l,m})$. RAB is applied by selecting uniformly distributed random phases $\theta_{R,m}(k) \sim \mbox{Unif}(0, 2\pi)$ and $\theta_{T,l}(k) \sim \mbox{Unif}(0, 2\pi)$. Because all phases assigned to the basis patterns are independent, then $y_{l}, l = 1,2,...,M_{T}$, are independent and identically distributed. 
It can be easily shown that when $\theta_{R,m}(k) \sim \mbox{Unif}(0,2\pi)$, then $\Psi = (\theta_{R,m}(k) + \theta_{T,l}(k) + \phi^{l,m}) \sim \mbox{Unif}(\theta_{T,l}(k) + \phi^{l,m}, \theta_{T,l}(k) + \phi^{l,m}+2\pi)$. Because we are only interested in $\Psi \, \mbox{mod} \, 2\pi$, it can be easily shown that $\Psi \, \mbox{mod} \, 2\pi \sim \mbox{Unif}(0, 2\pi)$. Using random variable transformation, the pdf of $y_{l,m}$ is given by
\begin{equation}
\label{pdfeq}
f_{y_{l,m}}(y_{l,m}) = \frac{1}{\pi \sqrt{1-y_{l,m}^{2}}}, -1 \leq y_{l,m} \leq 1,
\end{equation}
with $\mathbb{E}\{y_{l,m}\} = 0$, and $\mathbb{E}\{(y_{l,m}-\mathbb{E}\{y_{l,m}\})^{2}\} = \frac{1}{2}$. Applying the central limit theorem (CLT), we get $y_{l} = \sum_{m=1}^{M_{R}} y_{l,m} \sim \mathcal{N}(0,\frac{M_{R}}{2})$, which implies that $\sum_{l=1}^{M_{T}} y_{l} \sim \mathcal{N}(0,\frac{M_{R} M_{T}}{2})$. Therefore, $U^{R}(k) = A \sum_{l=1}^{M_{T}} y_{l} \sim \mathcal{N}(0,\frac{K \overline{\gamma}}{2(K+1)})$. It can be shown that $U^{I}(k) = {\bf \mathfrak{Im}}\{U(k)\}$ has the same distribution of $U^{R}(k)$, which implies that $U(k) \sim \mathcal{CN}\left(0,\frac{K \overline{\gamma}}{K+1}\right)$. Since $U(k) \sim \mathcal{CN}\left(0,\frac{K \overline{\gamma}}{K+1}\right)$, and $V(k) \sim \mathcal{CN}\left(0,\frac{\overline{\gamma}}{K+1}\right)$, then the equivalent channel after applying RAB will be $h^{eq}(k) = (V(k) + U(k)) \sim \mathcal{CN}(0,\overline{\gamma})$, which concludes the proof.

\section{Proof of Theorem 2}
\renewcommand{\theequation}{\thesection.\arabic{equation}}
We study each of the terms in (\ref{sch3}), and obtain their growth rates separately. Based on \cite[Lemma 2]{6}, we know that $\max_{n} \frac{\gamma_{s,n}}{\overline{\gamma}_{s}}-l_{s,N}$ converges in probability to a Gumbel distribution with a cumulative distribution function (cdf) of $\mbox{exp}(-e^{-\frac{\gamma_{s,n}}{K_{s}+1}})$, where $l_{s,N}$ is defined such that the cdf $F(l_{s,N}) = 1-\frac{1}{N}$, and is given by \cite[Eq. (5)]{6} as $l_{s,N} = \left(\sqrt{\frac{\log(N)}{K_{s}+1}} + \sqrt{\frac{K_{s}}{K_{s}+1}} \right)^{2} + \mathcal{O}(\log(\log(N)))$. Given that the mean of a Gumbel distribution is given by $l_{s,N} + \frac{\gamma}{K_{s}+1}$, the term $\frac{\mathbb{E}\left\{\max_{n} \gamma_{s,n} \right\}}{\mathbb{E}\left\{\gamma_{s,n} \right\}}$ grows like $\left(\sqrt{\frac{\log(N)}{K_{s}+1}} + \sqrt{\frac{K_{s}}{K_{s}+1}} \right)^{2} + \mathcal{O}(\log(\log(N))) + \mathcal{O}(1)$. As for the term $\frac{\mathbb{E}\left\{\max_{n} \frac{1}{\gamma_{sp,n}} \right\}}{\mathbb{E}\left\{ \frac{1}{\gamma_{sp,n}} \right\}}$, it can be written as $\frac{\mathbb{E}\left\{ \frac{1}{\min_{n} \gamma_{sp,n}} \right\}}{\mathbb{E}\left\{ \frac{1}{\gamma_{sp,n}} \right\}}$. We are thus interested in identifying the scaling behavior of $\min_{n} \gamma_{sp,n}$ when $h_{sp,n}$ follows a Rician distribution with a $K$-factor of $K_{sp}$. The cdf of $\gamma_{sp,n}$ is given by \cite{19}
\begin{equation}
\label{therm2}
F_{\gamma_{sp,n}}(\gamma_{sp,n}) = 1-Q_{1}\left(\sqrt{2K_{sp}}, \sqrt{\frac{2(K_{sp}+1)}{\overline{\gamma}_{sp}}\gamma_{sp,n}} \right),
\end{equation}
where $Q_{M}(a,b)$ is the Marcum-Q function. From \cite[Theorem 1]{27}, we know that the distribution of $x = (\min_{n} \gamma_{sp,n}-l_{sp,N})/d_{sp,N}$ converges to a Weibull distribution with $f_{x}(x) = 1-\mbox{exp}(-x^{\alpha})$ if $\lim_{t \to -\infty} \frac{F^{*}_{\gamma_{sp,n}}(t \gamma_{sp,n})}{F^{*}_{\gamma_{sp,n}}(t)} = \gamma_{sp,n}^{-\alpha}, \alpha > 0$, where $F_{z}^{*}(z) = F_{z}(\zeta(F_{z}(z))-\frac{1}{z})$, $\zeta(F_{z}(z)) = \mbox{inf}\{z: F_{z}(z)>0\}$, $l_{sp,N} = \zeta(F_{z}(z))$, and $d_{sp,N} = F_{z}^{-1}\left(\frac{1}{N}\right) - \zeta(F_{z}(z))$. It can be shown that 
\[\lim_{t \to \infty} \frac{1-Q_{1}\left(\sqrt{2K_{sp}}, \sqrt{\frac{2(K_{sp}+1)}{t\gamma_{sp,n}\overline{\gamma}_{sp}}} \right)}{1-Q_{1}\left(\sqrt{2K_{sp}}, \sqrt{\frac{2(K_{sp}+1)}{t \overline{\gamma}_{sp}}}\right)} =  \gamma_{sp}^{-1},\]       		
by replacing the Marcum-Q function with its series expansion [33, Eq. (4)] [26, Eq. (8.445)] 
\[Q_{1}\left(\sqrt{2K_{sp}}, \sqrt{\frac{2(K_{sp}+1)}{t\gamma_{sp,n}\overline{\gamma}_{sp}}} \right) = e^{-K_{sp}-\frac{K_{sp}+1}{\overline{\gamma}_{sp} \gamma_{sp,n} t }} \times \]
\[\sum_{v=0}^{\infty}\sum_{m=0}^{\infty}\frac{(K_{sp}\overline{\gamma}_{sp})^{\frac{v}{2}}  \left(\frac{K_{sp}(1+K_{sp})}{\overline{\gamma}_{sp}}\right)^{2m+v} }{\Gamma(v-1) (K_{sp}+1)^{v/2} \Gamma(m+v+1) (t\gamma_{sp})^{\frac{1}{2m+v/2}}},\]
and applying l'H\^{o}pital's rule once. From the above results, we have $\alpha = 1$, $l_{sp,N} = 0$, and $d_{sp,N} = F_{z}^{-1}\left(\frac{1}{N}\right)$. Thus, $\gamma_{sp, min} = \min_{n} \gamma_{sp,n}$ converges to an exponential distribution, i.e., $\lim_{N\to \infty} F_{\gamma_{sp, min}}(\gamma_{sp, min}) = 1-\mbox{exp}\left(-\frac{\gamma_{sp, min}}{\overline{\gamma}_{sp} d_{sp,N}}\right)$. We can obtain $d_{sp,N}$ by solving
\[Q_{1}\left(\sqrt{2K_{sp}}, \sqrt{\frac{2(K_{sp}+1)d_{sp,N}}{\overline{\gamma}_{sp}}} \right) = \frac{1}{N}\],
which can be evaluated by simple algebraic manipulation of \cite[Eq. (5)]{6} as
\[d_{sp,N} = \left(\sqrt{\frac{\log\left(\frac{N}{N-1}\right)}{K_{sp}+1}} +\sqrt{\frac{K_{sp}}{K_{sp}+1}}\right)^{2} + \mathcal{O}\left(\log\, \log\left(\frac{N}{N-1}\right)\right).\]
Because $\frac{1}{\gamma_{sp, min}}$ is a convex function, it follows from Jensen's inequality that $\mathbb{E}\left\{\frac{1}{\gamma_{sp, min}}\right\} \geq \frac{1}{\mathbb{E}\left\{\gamma_{sp, min}\right\}}$. Thus, $\mathbb{E}\left\{\frac{1}{\gamma_{sp, min}}\right\}$ grows at least as fast as $\frac{1}{\mathbb{E}\left\{\gamma_{sp, min}\right\}}$. As $\mathbb{E}\left\{\gamma_{sp, min}\right\} = \frac{1}{\overline{\gamma}_{sp} d_{sp,N}}$, the term $\frac{\mathbb{E}\left\{\frac{1}{\min_{n}\gamma_{sp,n}}\right\}}{\mathbb{E}\left\{\frac{1}{\gamma_{sp,n}}\right\}}$ grows at least as fast as $\frac{1}{d_{sp,N}}$. For large $N$, the term $\left(1-\frac{1}{N}\right) \to e^{\frac{-1}{N}}$. Hence, $\lim_{N \to \infty} d_{sp,N} = \left(\sqrt{\frac{1}{N(K_{sp}+1)}} +\sqrt{\frac{K_{sp}}{K_{sp}+1}}\right)^{2} + \mathcal{O}\left(\log(N)\right)$\textbf{($\log(1/N)??$)}. The same analysis can be applied for the PU-to-SU channel if $\overline{\gamma}_{p} >> 1$. This concludes the proof of the theorem.

\section{Proof of Theorem 3}
\renewcommand{\theequation}{\thesection.\arabic{equation}}
From (\ref{Th1}), we know that $h^{eq}_{s,n}(k) = U_{s,n}(k) + V_{s,n}(k)$. While the scattered component $V_{s,n}(k)$ follows a Rayleigh distribution, the dynamic range of the artificial fading component depends on $M_{T}$ and $M_{R}$. The largest possible absolute value of artificial fading term
\[|U_{s,n}(k)|=\]
\[\sqrt{\frac{K_{s} \overline{\gamma}_{s}}{(K_{s}+1)M_{T}M_{R}}} \, \left|\sum_{l=1}^{M_{T}} \sum_{m=1}^{M_{R}} e^{j(\theta_{R,m}(k)+\theta_{T,l}(k)+\phi_{s}^{l,m})}\right|\]
is $\sqrt{\frac{K_{s} \overline{\gamma}_{s} M_{T}M_{R}}{(K_{s}+1)}}$, when all the random phases are in the beamforming configuration with respect to the fixed component of the channel gain for a certain user. For a large number of users, i.e., $N \to \infty$, then there exists almost surely a fraction $\epsilon$ of users for which
\[\sqrt{\frac{K_{s} \overline{\gamma}_{s}}{(K_{s}+1)M_{T}M_{R}}} \, \left|\sum_{l=1}^{M_{T}} \sum_{m=1}^{M_{R}} e^{j(\theta_{R,m}(k)+\theta_{T,l}(k)+\phi_{s}^{l,m})}\right| \]
\[> \sqrt{\frac{K_{s} \overline{\gamma}_{s} M_{T}M_{R}}{(K_{s}+1)}} - \delta,\]
for any $\delta > 0$. This set of $\epsilon N$ users can be thought of as experiencing Rician fading with a fixed component magnitude close to $\sqrt{\frac{K_{s} \overline{\gamma}_{s} M_{T}M_{R}}{(K_{s}+1)}}$. Thus, the term $\frac{\mathbb{E}\left\{ \max_{n} \gamma^{eq}_{s,n} \right\}}{\mathbb{E}\left\{ \gamma^{eq}_{s,n} \right\}}$ in (\ref{sch3}) scales like
\begin{equation}
\label{Thm3eq1}
\left( \sqrt{\frac{\log(N)}{K_{s}+1}} +  \sqrt{\frac{M_{T}M_{R}K_{s}}{K_{s}+1}}\right)^{2} +\mathcal{O}(\log(\log(N))).
\end{equation}
For the SU-to-PU interference channel, we know that $h^{eq}_{sp,n}(k) = U_{sp,n}(k) + V_{sp,n}(k)$. The magnitude of the artificial fading component after RAB is
\[|U_{sp,n}(k)| = \sqrt{\frac{K_{sp} \overline{\gamma}_{sp}}{(K_{sp}+1)M_{T}}} \, \left|\sum_{l=1}^{M_{T}} e^{j(\theta_{T,l}(k)+\phi_{sp}^{l})}\right|.\]
For $N \to \infty$, a fixed $\delta > 0$, and $\epsilon \in (0,1)$, there exists almost surely a fraction $\epsilon$ of users such that
\[\sqrt{\frac{K_{sp} \overline{\gamma}_{sp}}{(K_{sp}+1)M_{T}}} \, \left|\sum_{l=1}^{M_{T}} e^{j(\theta_{T,l}(k)+\phi_{sp}^{l})}\right| < \delta.\]
Thus, these $\epsilon N$ can be considered as experiencing Rician fading with the magnitude of the fixed component equal to $\delta$. Among these $\epsilon N$, it follows from Theorem 1 that the term $\frac{\mathbb{E}\left\{\max_{n} \frac{1}{\gamma^{eq}_{sp,n}} \right\}}{\mathbb{E}\left\{\frac{1}{\gamma^{eq}_{sp,n}} \right\}}$ grows at least as fast as $ \frac{1}{\left( \sqrt{\frac{1}{\epsilon N(K_{sp}+1)}} +  \delta \right)^{2}}$.
For infinitesimally small $\delta$, the term $\frac{\mathbb{E}\left\{\max_{n} \frac{1}{\gamma^{eq}_{sp,n}} \right\}}{\mathbb{E}\left\{\frac{1}{\gamma^{eq}_{sp,n}} \right\}}$ grows linearly with $N$. The same analysis can be applied on the PU-to-SU after negelecting the noise power at the SU receiver. Combining these results, the growth rate of $\overline{\gamma}_{PAC}(N)$ after applying RAB can be easily evaluated.

\end{document}